\documentclass[12pt,letterpaper]{article}
\usepackage[T1]{fontenc}
\usepackage{lmodern}
\usepackage[margin=1in]{geometry}
\usepackage{amsmath,amssymb,amsthm,mathtools}
\usepackage{booktabs,tabularx,array,longtable}
\usepackage{graphicx}
\usepackage{placeins}
\usepackage{xcolor}
\usepackage{microtype}
\usepackage{enumitem}
\usepackage{algorithm}
\usepackage{algorithmic}
\usepackage[authoryear,round]{natbib}
\usepackage[font=small,labelfont=bf]{caption}
\usepackage{setspace}
\usepackage{url}
\usepackage[hidelinks]{hyperref}
\AtBeginDocument{\DeclareFontShape{T1}{lmr}{bx}{sc}{<->ssub*lmr/m/sc}{}}
\hypersetup{pdfauthor={Zijian Wei, Yajie Bao, Haojie Ren, Zhen Chen},
  pdftitle={FissionGuard: Testing Local Watermark Proportions with FDR Control}}

\newcommand{\method}{\textsc{FissionGuard}}
\newcommand{\FDR}{\operatorname{FDR}}
\newcommand{\FDP}{\operatorname{FDP}}
\newcommand{\Unif}{\operatorname{Unif}}
\newcommand{\E}{\mathbb E}
\newcommand{\PP}{\mathbb P}
\newcommand{\1}{\mathbf 1}
\newcolumntype{Y}{>{\raggedright\arraybackslash}X}

\theoremstyle{plain}
\newtheorem{theorem}{Theorem}[section]
\newtheorem{lemma}[theorem]{Lemma}
\newtheorem{proposition}[theorem]{Proposition}
\newtheorem{corollary}[theorem]{Corollary}
\theoremstyle{definition}
\newtheorem{assumption}[theorem]{Assumption}

\theoremstyle{remark}

\theoremstyle{definition}

\title{\method: Testing Local Watermark \\ Proportions with FDR Control}
\author{Zijian Wei$^1$ \quad Yajie Bao$^2$ \quad Haojie Ren$^{1,}$\thanks{\raggedright Corresponding authors: Haojie Ren (\texttt{haojieren@sjtu.edu.cn}) and Zhen Chen (\texttt{chenzhendr@sjtu.edu.cn}).} \quad Zhen Chen$^{3,*}$\\[6pt]
$^1$ {\small School of Mathematical Sciences, Shanghai Jiao Tong University}\\
$^2$ {\small School of Management, Fudan University}\\
$^3$ {\small Department of Industrial Engineering \& Management, Shanghai Jiao Tong University}}
\date{\today}
\begin{document}
\maketitle

\begin{abstract}
A watermark detector can flag a text even when only a small fraction of its tokens are watermarked. We study sequential auditing to identify local passages whose watermark proportions exceed a prespecified threshold. Each alarm fixes a candidate passage's endpoint, while its start is chosen from earlier tokens. Reusing the same evidence to locate and test a passage can invalidate inference. We propose \method{}, which uses data fission to construct an alarm value and a validation value from each token-level $p$-value. The alarm values guide alarms and localization. The validation values yield passage-level $p$-values through partial conjunction, excluding as many of the strongest signals as the proportion null allows. These $p$-values are used to test the corresponding hypotheses at levels assigned by an online multiple-testing procedure. Under the stated watermark generation model, allowing dependence among token-level $p$-values at watermarked positions, we prove finite-sample false discovery rate (FDR) control among reported passages at every fixed finite horizon. We also derive a finite-sample lower bound on power for testing nonnull candidates. Experiments on synthetic streams and Mistral and Qwen generations across four marking layouts show gains in watermarked-token coverage over sample splitting, with empirical FDR below the target.
\end{abstract}

\section{Introduction}\label{sec:introduction}

As large language models (LLMs) are increasingly used to generate text, it becomes more important to identify the source of generated content and track its usage. Watermarking has attracted substantial attention as a means of providing statistical evidence of text source \citep{kirchenbauer2023watermark,kuditipudi2024robust,hu2024unbiased}. By linking generation to secret randomness, watermark allows us to test whether the text carries evidence of the corresponding generation mechanism. Texts encountered in practice often combine human writing, model outputs, and subsequent edits. Therefore, in {\it mixed-source text}, unwatermarked content may occur both between and within watermarked passages. A long passage containing a short watermarked insertion may provide strong evidence that the target watermark is present, even when most of its content is unwatermarked. This raises two questions: how to find those passages containing watermark evidence, and how to establish that their watermark proportions exceed a prespecified threshold?

Recent works have studied watermark detection, localization and watermark proportions. \citet{li2026robust} construct robust global tests to detect watermark signals under human edits.  Localization methods identify candidate passages \citep{pan2025waterseeker,zhao2025segments,bonnerjee2026wiser}, and \citet{blanchet2026localization} study token-level multiple testing with asymptotic error rate control. Meanwhile, \citet{li2025proportion} estimate watermark proportions through mixture models for pivotal statistics. For online detection, \citet{su2026watermark} construct e-processes for anytime-valid tests of watermark presence, and \citet{xu2026lordgof} apply online FDR control to document-level goodness-of-fit tests. These works provide foundations for localizing watermarks, but they do not directly establish proportion claims for passages selected adaptively from online evidence.

This work aims to test whether such a passage's watermark proportion exceeds a prespecified threshold, while controlling the false discovery rate (FDR) among reported passages. Reusing the evidence that located a passage can invalidate a test calibrated for a fixed passage \citep{fithian2014optimal,jewell2022changepoint}.
To address this issue, we propose \method{}, which separates passage localization from validation through data fission \citep{leiner2025fission}. As illustrated in Figure~\ref{fig:framework}, each token-level $p$-value is transformed into an alarm value and a validation value.
The alarm values determine when an alarm occurs and locate the candidate passage's starting point. The validation values are reserved until the candidate passage is fixed. Every token can therefore contribute to both tasks, with its information divided between the two values. We then introduce partial-conjunction testing \citep{benjamini2008screening,benjamini2009selective} to construct a passage-level $p$-value that accounts for the watermarked positions allowed under the proportion null. An online multiple-testing procedure \citep{javanmard2018online,ramdas2017decaying,ramdas2018saffron} is applied to the resulting $p$-values to control FDR and determines whether to report the passage.   %

\begin{figure}[!htb]
\centering
\includegraphics[width=\linewidth]{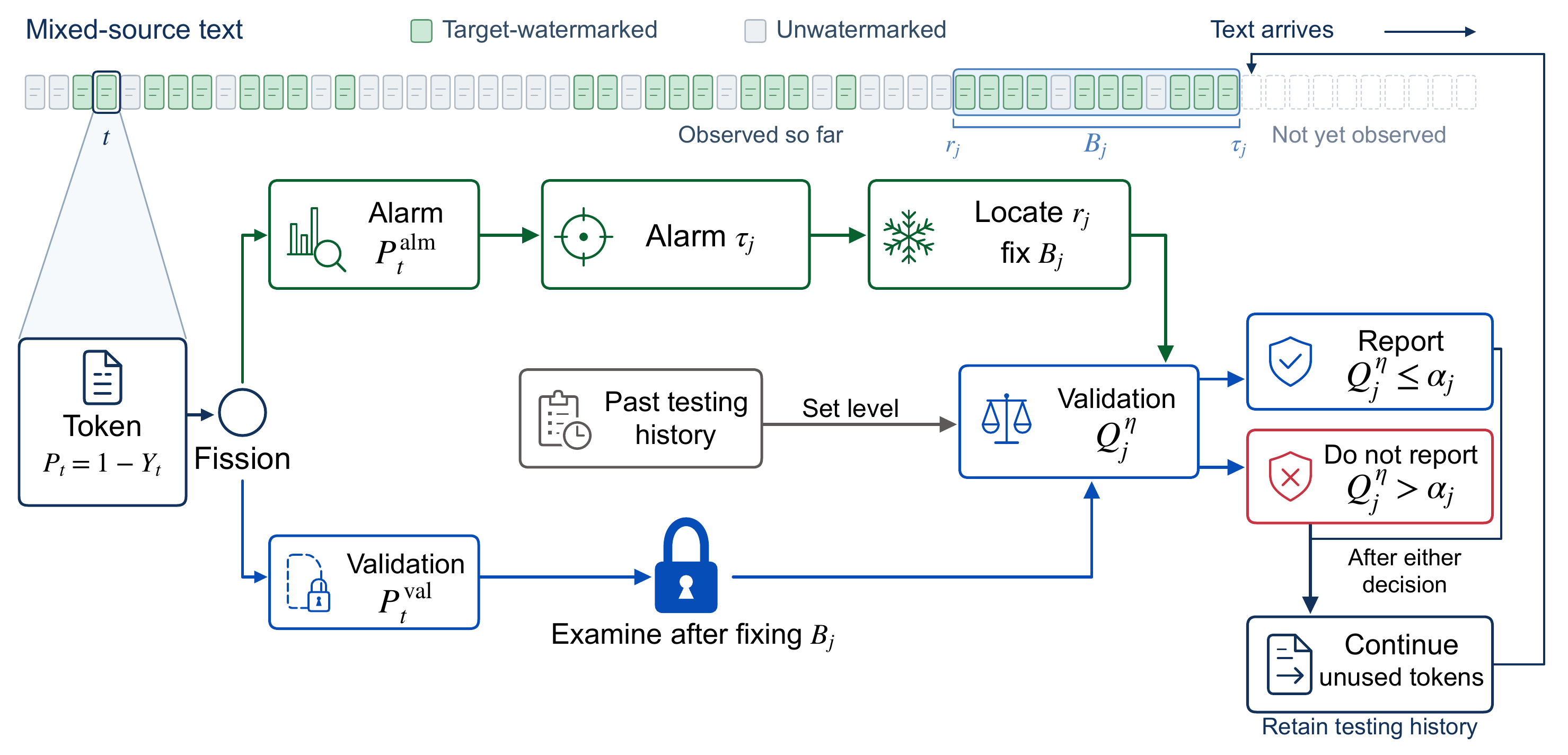}
\caption{\textbf{Locate a candidate passage, then test its proportion null.} Alarm values determine the endpoint $\tau_j$ and starting boundary $r_j$ of a candidate passage $B_j$. Validation values yield the passage-level $p$-value $Q_j^\eta$. After either testing decision, the detector restarts while the online testing history is retained.}
\label{fig:framework}
\end{figure}

The main contributions of this work are summarized as follows:
\begin{itemize}
    \item We develop \method{} to test watermark proportions in candidate passages chosen as tokens arrive, addressing both passage localization and the watermarked positions permitted under the proportion null. 
    \item We establish calibration after selection and finite-sample FDR control among reported passages at every fixed finite horizon, allowing dependence among token-level $p$-values at watermarked positions. We also derive a finite-sample lower bound on power among nonnull candidates. 
    \item Experiments on text generated by Mistral and Qwen show that \method{} achieves high power, token coverage, and qualified coverage, with empirical FDR below the target. The results also illustrate the coverage--purity trade-off.
\end{itemize}

\section{Problem formulation}\label{sec:setup}

\subsection{Statistical framework for watermark detection}
\label{sec:token-evidence}

Let $W_t\in\mathcal V$ denote the token at position $t$, where $\mathcal V$ is the vocabulary. Given the prompt and previously generated tokens, a language model specifies a next-token prediction (NTP) distribution $\pi_t=(\pi_{t,w})_{w\in\mathcal V}$, where $\pi_{t,w}$ is the probability assigned to token $w$. Without watermarking, the next token $W_t$ is sampled directly from $\pi_t$. With watermarking, the model introduces a pseudorandom variable $\zeta_t$ and generates the next token as \(W_t=\mathcal S(\pi_t,\zeta_t)\), where $\mathcal S$ is the watermark decoding function. For example, in the Gumbel-max watermark \citep{aaronson2023watermarking}, the pseudorandom vector $\zeta_t=(U_{t,w})_{w\in\mathcal V}$ consists of $|\mathcal V|$ independent random variables sampled from $\Unif(0,1)$, then the next token is output by $W_t=\arg\max_{w\in \mathcal V}\; \log (U_{t,w})/\pi_{t,w}$.

Following \citet{li2025statistical}, the evidence is summarized by a pivotal statistic $Y_t=Y(W_t,\zeta_t)$ such that $Y(w,\zeta_t)$ has the same known continuous distribution function $F_0$ for every fixed token $w\in\mathcal V$. For the unwatermarked output, we have $Y_t\sim F_0$ regardless of the NTP distribution since $W_t$ and $\zeta_t$ are independent. With larger values indicating stronger watermark evidence, the token-level $p$-value is $P_t=1-F_0(Y_t)$. At unwatermarked positions, $P_t$ is uniform conditional on the preceding history. For example, the Gumbel-max watermark uses $Y_t=U_{t, W_t}$ with $F_0$ being $\Unif(0,1)$. These values provide the inputs for passage-level testing.

\subsection{Passage hypotheses and FDR}
\label{sec:interval-hypotheses}

Let $\mathcal M$ denote the unknown set of positions generated using the target watermark decoder $\mathcal S$, and let $\mathcal N$ contain the remaining positions. Those unwatermarked positions may occur both between and within watermarked passages.

As tokens arrive, candidate passages are identified through two components:
\begin{itemize}
    \item $\mathsf{Alarm}$:
    An online rule that processes token-level evidence as it arrives
    and raises alarms at ordered times $\tau_1,\tau_2,\ldots$.
    The $j$th alarm fixes the endpoint $\tau_j$ of a candidate passage.

    \item $\mathsf{Locate}$:
    At each alarm time $\tau_j$, a localization rule selects a starting
    boundary $0\le r_j<\tau_j$ using evidence available by that time.
    The candidate passage consists of tokens
    $W_{r_j+1},\ldots,W_{\tau_j}$.
\end{itemize}
Here, the alarm rule runs until an alarm occurs, and we restart the detector at the next unused token whether or not the candidate passage is reported. Restarting after an alarm is standard in online monitoring with multiple changes \citep{yu2023onlinechangepoint,du2023dynamic}. Section~\ref{sec:separation} explains how we separate the evidence used for alarm and localization from that used to test the candidate.

For the candidate selected at $\tau_j$, define its token-index interval
and the maximum number of watermarked positions allowed under the null by
\begin{equation}
B_j=\{r_j+1,\ldots,\tau_j\},
\qquad
b_j=\lfloor\eta|B_j|\rfloor,
\label{eq:report-interval}
\end{equation}
where \(\eta\in[0,1)\) is a proportion threshold fixed before monitoring. To validate these passages, we formulate an online multiple-testing problem:
\begin{equation}
H_{0j}^{\eta}:|B_j\cap\mathcal M|\le b_j,
\qquad
H_{1j}^{\eta}:|B_j\cap\mathcal M|>b_j.
\label{eq:interval-null}
\end{equation}
A report asserts that the watermark proportion of $B_j$ exceeds $\eta$. Under the null, at most \(b_j\) positions may be watermarked, no matter how strong their evidence is. At $\eta=0$, it reduces to the unwatermarked null on that passage. 

We assess reliability across reported passages through the false discovery rate (FDR) \citep{benjamini1995fdr,javanmard2018online}.
Let \(R_j\in\{0,1\}\) indicate whether the procedure reports \(B_j\) at time \(\tau_j\), using only information available by that time. At every fixed finite token horizon $T$ and a specified level $\alpha$, the target is
\begin{equation}
\FDR_\eta(T)
=
\E\!\left[
\frac{
\sum_{j:\tau_j\le T}
R_j\1\{|B_j\cap\mathcal M|\le b_j\}
}{
1\vee\sum_{j:\tau_j\le T}R_j
}
\right]
\le\alpha.
\label{eq:alarm-fdr}
\end{equation}
This bounds the expected fraction of reported passages whose watermark proportions do not exceed $\eta$.

Two distinct issues arise when testing these candidate passages. First, token-level $p$-values at unwatermarked positions may not remain valid after a passage is selected using those values. For example, let $P_t$ be independent uniform token-level $p$-values and stop at $\tau=\inf\{t\ge1:P_t\le h\}$ for a fixed $h\in(0,1)$. Then $\PP(P_\tau\le x\mid\tau<\infty)=x/h>x$ for $0<x\le h$, so $P_\tau$ is not a valid $p$-value after selection. Second, when $\eta>0$, a passage satisfying $H_{0j}^\eta$ can contain watermarked tokens with strong evidence. A presence test may therefore reject without establishing that the watermark proportion exceeds $\eta$.

\section{\method{}}\label{sec:method}

\method{} separates passage selection from validation, constructs a test for each candidate passage, and controls errors across reports. The alarm and localization rules use only the alarm values. Validation values are reserved for testing, and completed tests determine subsequent testing levels.

\subsection{Separating alarm from validation via data fission}\label{sec:separation}
Motivated by data fission \citep{leiner2025fission}, we construct an alarm value and a validation value from each token-level $p$-value $P_t$, defined in Section \ref{sec:token-evidence}. Choose $\rho\in(0,1)$, draw $V_t\sim N(0,1)$ independently of the token-level $p$-values, and let $\Phi$ be the standard normal distribution function. Define
\begin{equation}
P_t^{\sf alm}=\Phi\!\left(\sqrt\rho\,\Phi^{-1}(P_t)-\sqrt{1-\rho}\,V_t\right),\quad
P_t^{\sf val}=\Phi\!\left(\sqrt{1-\rho}\,\Phi^{-1}(P_t)+\sqrt\rho\,V_t\right).\label{eq:fission}
\end{equation}
The alarm and localization rules use only $P_t^{\sf alm}$, while $P_t^{\sf val}$ is examined after its candidate passage is fixed. Smaller values indicate stronger watermark evidence in both sequences. Figure~\ref{fig:selection-calibration} illustrates this separation. Under the conditions stated below, validation values at unwatermarked positions remain uniform after passage selection, addressing the first issue in Section~\ref{sec:interval-hypotheses}.
If $P_t \sim \Unif(0,1)$, \eqref{eq:fission} is an orthogonal Gaussian rotation. Larger $\rho$ increases the information available for alarms and localization and decreases it in validation, and $\rho=1/2$ gives a symmetric allocation.

\begin{figure}[!htb]
\centering
\includegraphics[width=0.95\linewidth]{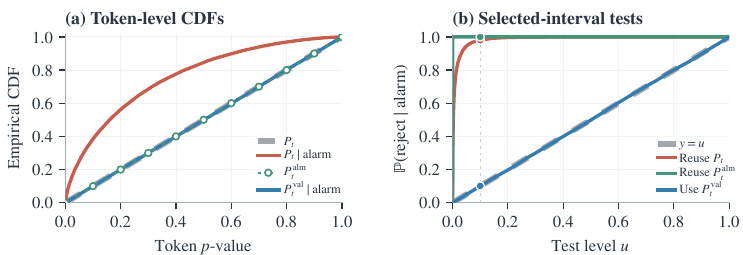}
\caption{\textbf{Evidence reuse after a false alarm.} Illustration on unwatermarked synthetic data. (a) Token-level CDFs before and after selection. (b) Rejection probabilities under the null conditional on an alarm.}
\label{fig:selection-calibration}
\end{figure}

We analyze the statistical watermark detection framework introduced in Section~\ref{sec:token-evidence}. Before that, we impose the following assumption.

\begin{assumption}[Unused watermark randomness]\label{ass:interface}
The pseudorandom variables $(\zeta_t)_{t\in\mathcal N}$ at unwatermarked positions are mutually independent, with each $\zeta_t$ following the distribution specified by the watermarking scheme. Moreover,
\begin{equation}
(\zeta_t)_{t\in\mathcal N}\perp
\bigl(W_{1:\infty},(\zeta_t)_{t\in\mathcal M}\bigr).
\label{eq:unused-randomness}
\end{equation}
\end{assumption}
Assumption~\ref{ass:interface} restricts the pseudorandom variables at unwatermarked positions while allowing dependence in the text and among pseudorandom variables at watermarked positions. Appendix~\ref{app:watermark-input} discusses the shared conditions in earlier watermark models, gives a sufficient generation model, and proves the resulting joint distribution of the token-level $p$-values.

\begin{proposition}[Conditional distribution of validation values]\label{thm:factorization}
Under Assumption~\ref{ass:interface} and independent fission noise,
\begin{equation}
P_t^{\sf val}\mid\bigl(P_{1:\infty}^{\sf alm},P_{\mathcal M}^{\sf val}\bigr)\overset{\rm iid}{\sim}\Unif(0,1),\qquad t\in\mathcal N.
\label{eq:complete-separation}
\end{equation}
\end{proposition}
The proposition gives the distribution of the reserved validation values after conditioning on the complete alarm sequence. Conditioning also on marked validation values accounts for their possible presence inside a passage satisfying $H_{0j}^\eta$. Appendix~\ref{app:watermark-input} gives the proof. 

\subsection{From an alarm to a reported passage}\label{sec:reporting}
An online detector raises an alarm, and a localization rule chooses the start of the candidate passage. Let $\mathsf{Alarm}$ be a measurable binary function of the alarm values observed since the most recent restart and the declared request boundaries. It resets after each alarm or external request end. Its ordered alarm times satisfy
\begin{equation}
\tau_j=\inf\{t>\tau_{j-1}:\mathsf{Alarm}(P_{1:t}^{\sf alm})=1\},\qquad\tau_0=0,\quad\inf\varnothing=\infty.\label{eq:alarm-time}
\end{equation}
Thus $\{\tau_j\le t\}$ is measurable with respect to the available alarm information. CUSUM \citep{page1954inspection}, FOCuS \citep{romano2023focus} and e-process detectors \citep{ramdas2023savi} can serve as alarm rules. The validity argument does not require a particular detector. Our watermark experiments use a Stouffer scan. Its settings are given in Appendix~\ref{app:implementation}.

At an alarm, $\mathsf{Locate}$ chooses the candidate's starting boundary using a measurable localization score:
\begin{equation}
r_j=\left[\underset{1\le s\le\tau_j}{\arg\max}\;\ell_j(s;P_{1:\tau_j}^{\sf alm})\right]-1.\label{eq:localizer}
\end{equation}
Starts before the most recent restart or yielding lengths outside the allowed range receive score $-\infty$. The implementation uses the positive standardized sum of normal scores (Appendix~\ref{app:implementation}). These choices fix $\tau_j$, $r_j$, $B_j$ and $b_j$ before using the validation values. After either testing decision, the detector restarts at the next token. Each candidate passage uses only tokens observed since the most recent restart, so candidate passages are disjoint and no validation value is reused.

The null in~\eqref{eq:interval-null} is a partial-conjunction null \citep{benjamini2008screening,benjamini2009selective}: it allows at most $b_j$ watermarked positions, while the alternative requires at least $b_j+1$. For each deterministic nonempty subset $I$, fix a combination rule $\mathcal Q_I$ calibrated on independent uniforms. Define the passage-level $p$-value for $H_{0j}^\eta$ by
\begin{equation}
Q_j^\eta=\max_{\substack{I\subseteq B_j\\|I|=|B_j|-b_j}}\mathcal Q_I(P_I^{\sf val}).
\label{eq:confirmation-interface}
\end{equation}
This partial-conjunction construction \citep{wang2019partial} uses the fact that, under the null, at least one subset in this maximum is entirely unwatermarked. Once its validation coordinates retain their independent uniform law, the maximum dominates a combined $p$-value on that subset and is therefore conservative. Section~\ref{sec:separation} establishes the required conditional law. The combination rule is fixed before validation is examined. Appendix~\ref{app:implementation} gives the implementation.

Let $\mathcal A$ be an online multiple-testing algorithm at target level $\alpha$, chosen from SAFFRON \citep{ramdas2018saffron}, LORD \citep{javanmard2018online}, and LORD++ \citep{ramdas2017decaying}, with their respective parameters. Before using the current validation values, $\mathcal A$ assigns $\alpha_j$ using only completed tests. The reporting indicator is
\begin{equation}
R_j=\1\{Q_j^\eta\le\alpha_j\}.
\label{eq:report-decision}
\end{equation}
Every completed test updates the same testing history, including tests that do not report. At an external request end, the detector is reset without creating a candidate from the remaining tokens. The testing history and the candidate counter are retained.

\begin{algorithm}[!htb]
\caption{\method{} for a fixed proportion threshold $\eta$}\label{alg:fissionguard}
\begin{algorithmic}[1]
\REQUIRE Token-level $p$-values $P_t$, allocation $\rho$, threshold $\eta$, alarm and localization rules, calibrated combinations $\mathcal Q_I$, online algorithm $\mathcal A$ at target $\alpha$.
\STATE Initialize the detector, candidate index $j=1$, and $\mathcal A$ with empty testing history.
\FOR{each arriving token $t$}
\STATE Form $P_t^{\sf alm}$ and $P_t^{\sf val}$ by~\eqref{eq:fission}. Use only $P_t^{\sf alm}$ for alarm and localization.
\IF{the alarm rule triggers as in~\eqref{eq:alarm-time}}
\STATE Fix $\tau_j$, $r_j$, $B_j$ and $b_j$ using only the alarm values.
\STATE Use $\mathcal A$ to assign $\alpha_j$ from completed tests, then compute $Q_j^\eta$ from $P_{B_j}^{\sf val}$.
\STATE Report $B_j$ if $Q_j^\eta\le\alpha_j$, and update $\mathcal A$ after either decision.
\STATE Increment $j$ and restart the detector at the next unused token.
\ENDIF
\STATE At an external request end, reset the detector without creating a candidate from the remaining tokens; retain the testing history and $j$.
\ENDFOR
\ENSURE Alarm and localization rules do not read validation values or test outcomes, and validation coordinates are not reused.
\end{algorithmic}
\end{algorithm}

\begin{theorem}[FDR control for reported passages]\label{cor:online-fdr}
Under Assumption~\ref{ass:interface} and independent fission noise, let Algorithm~\ref{alg:fissionguard}, with $\mathcal A$ chosen as above, use alarm and localization rules that depend only on the alarm values and fixed external design. For each true $H_{0j}^\eta$,
\begin{equation}
\PP\!\left(Q_j^\eta\le u\mid P_{1:\infty}^{\sf alm},P_{\mathcal M}^{\sf val},Q_{1:j-1}^\eta\right)\le u,\qquad 0\le u\le1.
\label{eq:selected-calibration}
\end{equation}
Consequently, the procedure above satisfies $\FDR_\eta(T)\le\alpha$ at every fixed finite token horizon $T$.
\end{theorem}
Conditioning fixes the candidate passages, their labels and their number by time $T$. On a passage satisfying $H_{0j}^\eta$, at least $|B_j|-b_j$ coordinates are unwatermarked, so the maximum in~\eqref{eq:confirmation-interface} dominates a calibrated combination on a subset consisting only of unwatermarked positions. Disjoint passages give conditional independence even when some marked coordinates lie in a true null. Applying the FDR guarantee for independent p-values of $\mathcal A$ under this conditioning and then averaging proves the claim. Appendix~\ref{app:proofs} gives the full proof. At $\eta=0$, the maximum has one term and validation reduces to the full combination.

\subsection{Power analysis}\label{sec:power}
Beyond FDR control, we also analyze power to assess how often nonnull candidate passages are reported. For each candidate produced by time $T$, let $D_j^\eta=\1\{|B_j\cap\mathcal M|>b_j\}$. We measure power among candidate passages satisfying $H_{1j}^\eta$ by
\begin{equation}
\mathrm{Pow}_\eta(T)=\frac{\E\left[\sum_{j:\tau_j\le T}D_j^\eta R_j\right]}{\E\left[\sum_{j:\tau_j\le T}D_j^\eta\right]}.
\label{eq:conditional-power-target}
\end{equation}
When the denominator is zero, we set $\mathrm{Pow}_\eta(T)=0$. Power measures successful tests among the method's nonnull candidate passages. We report token coverage separately to measure recovery across all watermarked positions. For each deterministic interval $B$, define $Q_B^\eta$ from its validation coordinates whether or not it becomes a candidate.

\begin{theorem}[A lower bound on power]\label{thm:main-power}
Assume $\E\left[\sum_{j:\tau_j\le T}D_j^\eta\right]>0$. Let $\mathcal B_T$ be a deterministic finite family of nonempty intervals in $\{1,\ldots,T\}$. Suppose every testing level assigned by Algorithm~\ref{alg:fissionguard} through $T$ is at least a deterministic $u\in(0,1)$. Then
\begin{equation}
\mathrm{Pow}_\eta(T)\ge\left[1-
\frac{\E\left[\sum_{j:\tau_j\le T}D_j^\eta\1\{B_j\notin\mathcal B_T\}\right]
+\sum_{B\in\mathcal B_T}\PP\{Q_B^\eta>u\}}
{\E\left[\sum_{j:\tau_j\le T}D_j^\eta\right]}\right]_+.
\label{eq:main-power}
\end{equation}
\end{theorem}
The first numerator term counts nonnull candidates outside $\mathcal B_T$, and the second bounds failures to reject within it. The proof uses a counting argument. For the partial-conjunction test based on Stouffer's method, Lemma~\ref{lem:confirmation-tail} in Appendix~\ref{app:power-main-proof} bounds $\PP\{Q_B^\eta>u\}$ using capped normal scores.

\begin{corollary}[An exponential lower bound on power]\label{cor:exponential-power}
In the setting of Theorem~\ref{thm:main-power}, suppose the conditions on Gumbel-max generation and signal strength in Appendix~\ref{app:power-statements} hold. There are constants $C,c_1>0$ such that, if every $B\in\mathcal B_T$ has length at least $n_{\min}$ and $n_{\min}\ge C\{1+\log(1/u)\}$, then
\begin{equation}
\mathrm{Pow}_\eta(T)\ge\left[1-
\frac{\E\left[\sum_{j:\tau_j\le T}D_j^\eta\1\{B_j\notin\mathcal B_T\}\right]
+|\mathcal B_T|e^{-c_1n_{\min}}}
{\E\left[\sum_{j:\tau_j\le T}D_j^\eta\right]}\right]_+.
\label{eq:exponential-power-bound}
\end{equation}
Moreover, $\mathrm{Pow}_\eta(T)\to1$ if all nonnull candidates satisfy the same signal and length conditions and have length at least a sufficiently large multiple of $\log T$, their expected count is bounded away from zero, and testing levels are bounded below by an inverse polynomial in $T$.
\end{corollary}
The corollary gives exponential decay in passage length for the nonrejection probability on intervals satisfying the signal condition, while allowing unwatermarked positions within those intervals. Appendix~\ref{app:power} gives the full conditions, proofs, and the inverse-transform analogue; Appendix~\ref{app:theory-diagnostic} evaluates the finite-sample bound.

\section{Experiments}\label{sec:experiments}
We compare \method{} with Splitting, e-process restart \citep{su2026watermark}, and Raw-p restart at $\alpha=0.1$, with default proportion threshold $\eta=0.25$ and $\rho=1/2$ for \method{}. \method{} and Splitting share the alarm and localization rules and test passage hypotheses using partial-conjunction tests and SAFFRON \citep{ramdas2018saffron}. Splitting uses odd-indexed tokens for alarms and localization and even-indexed tokens for validation. Its validation test targets the same full-passage hypothesis and uses the budget $b_j=\lfloor\eta|B_j|\rfloor$, since all watermarked positions allowed under the null could lie in the validation half. e-process restart reports complete crossing episodes, while Raw-p restart reports every candidate selected from unsplit token-level $p$-values. Neither restart baseline provides the same proportion guarantee. Alarms are at least 64 tokens apart, without imposing this minimum on localized passage lengths. Appendix~\ref{app:implementation} gives the implementation details. Figure~\ref{fig:watermark-screening} illustrates the resulting \method{} procedure on a mixed-watermark Mistral excerpt.

\begin{figure}[!htb]
\centering\includegraphics[width=0.95\linewidth]{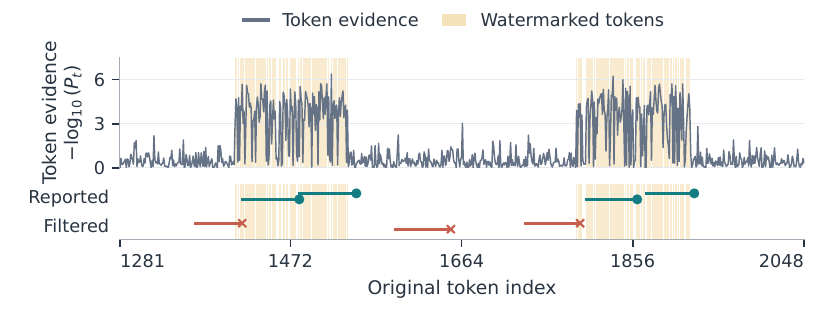}
\caption{\textbf{Candidate passages and validation within mixed watermark regions.} An illustrative Mistral-7B excerpt at $\eta=0.25$, containing two 128-token regions with 75\% watermarked positions. Shading marks the actual watermarked tokens. All seven candidates in the displayed range are shown: four nonnull candidates are reported and three proportion-null candidates are filtered.}\label{fig:watermark-screening}
\end{figure}

We report empirical FDR, empirical power estimating \eqref{eq:conditional-power-target}, and token coverage and purity of the report union. ESR instead measures restart success on complete nonnull episodes. Coverage counts watermarked tokens recovered by all reports, whereas qualified coverage, Cov.$^+$, counts those covered by reports whose watermark proportion exceeds $\eta$. Purity is the watermark proportion of the report union. Appendix~\ref{app:metric-conventions} gives the formula definitions.

\subsection{Controlled simulations and threshold sensitivity}\label{sec:synthetic}
The simulations separate evidence strength from the positions carrying it. At watermarked positions, the token-level $p$-values follow the exact uniform-vocabulary Gumbel-max law $P_t\sim\operatorname{Beta}(1,K)$, with $K=16$ in the main table, while all other token-level $p$-values are independent uniforms. Each deployment contains four regions with 75\% watermarked positions. The four layouts use the same region lengths, spacing, and number of watermarked tokens. Table~\ref{tab:synthetic} compares $\eta=0$, $0.25$ and $0.5$ on the same 200 independent deployments per layout. Appendix~\ref{app:matched-layouts} specifies the layouts and gives $\eta=0.1$ and the $K=64$ strength comparison.

\begin{table}[!htb]\centering
\caption{\textbf{Controlled online reporting at fixed signal strength, $K=16$, $\alpha=0.1$.} Each layout has 200 independent deployments with four 512-token regions, each 75\% marked. Thresholds share source streams and randomization. Dashes denote inapplicable metrics or undefined purity. Appendix~\ref{app:matched-layouts} gives the complete design and additional comparisons. Bold marks the highest displayed Power, ESR, coverage and purity among methods with FDR $\le\alpha$.}\label{tab:synthetic}
\fontsize{8}{9.5}\selectfont
\setlength{\tabcolsep}{0.6pt}
\renewcommand{\arraystretch}{1.06}
\begin{tabular*}{\linewidth}{@{\extracolsep{\fill}}l@{\hspace{4pt}}*{4}{r}@{\hspace{4pt}}*{4}{r}@{\hspace{4pt}}*{4}{r}@{}}\toprule
 & \multicolumn{4}{c}{\textbf{$\eta=0$}} & \multicolumn{4}{c}{\textbf{$\eta=0.25$}} & \multicolumn{4}{c}{\textbf{$\eta=0.5$}}\\
\cmidrule(lr){2-5}\cmidrule(lr){6-9}\cmidrule(lr){10-13}
Method & FDR & \shortstack{Power/\\ESR} & Cov. & Purity & FDR & \shortstack{Power/\\ESR} & Cov. & Purity & FDR & \shortstack{Power/\\ESR} & Cov. & Purity\\\midrule
\multicolumn{13}{c}{\textit{Contiguous}}\\[2pt]
\textbf{FissionGuard} & 0.0436 & \textbf{0.974}/-- & 0.986 & 0.872 & 0.0002 & \textbf{0.931}/-- & \textbf{0.966} & 0.952 & 0.0000 & \textbf{0.849}/-- & \textbf{0.871} & \textbf{0.983}\\[1pt]
Splitting & 0.0393 & \textbf{0.974}/-- & 0.994 & \textbf{0.879} & 0.0000 & 0.890/-- & 0.955 & \textbf{0.964} & 0.0000 & 0.000/-- & 0.000 & --\\[1pt]
e-process & 0.0039 & --/\textbf{1.000} & \textbf{1.000} & 0.223 & 0.1560 & --/1.000 & 1.000 & 0.223 & 0.1911 & --/1.000 & 1.000 & 0.223\\[1pt]
Raw-p & 0.2188 & --/1.000 & 0.999 & 0.687 & 0.2587 & --/1.000 & 0.999 & 0.687 & 0.3202 & --/1.000 & 0.999 & 0.687\\[1pt]
\addlinespace[3pt]
\multicolumn{13}{c}{\textit{Random}}\\[2pt]
\textbf{FissionGuard} & 0.0354 & \textbf{0.984}/-- & 0.973 & 0.680 & 0.0002 & \textbf{0.906}/-- & \textbf{0.922} & 0.734 & 0.0000 & \textbf{0.002}/-- & \textbf{0.003} & \textbf{0.832}\\[1pt]
Splitting & 0.0325 & 0.979/-- & 0.973 & \textbf{0.681} & 0.0000 & 0.623/-- & 0.660 & \textbf{0.759} & 0.0000 & 0.000/-- & 0.000 & --\\[1pt]
e-process & 0.0021 & --/\textbf{1.000} & \textbf{1.000} & 0.220 & 0.1289 & --/1.000 & 1.000 & 0.220 & 0.2033 & --/1.000 & 1.000 & 0.220\\[1pt]
Raw-p & 0.1744 & --/1.000 & 0.992 & 0.566 & 0.2065 & --/1.000 & 0.992 & 0.566 & 0.2754 & --/1.000 & 0.992 & 0.566\\[1pt]
\addlinespace[3pt]
\multicolumn{13}{c}{\textit{Gapped}}\\[2pt]
\textbf{FissionGuard} & 0.0392 & \textbf{0.970}/-- & 0.985 & 0.805 & 0.0005 & \textbf{0.878}/-- & \textbf{0.942} & 0.906 & 0.0000 & \textbf{0.680}/-- & \textbf{0.726} & \textbf{0.976}\\[1pt]
Splitting & 0.0313 & 0.962/-- & 0.991 & \textbf{0.819} & 0.0000 & 0.795/-- & 0.909 & \textbf{0.930} & 0.0000 & 0.000/-- & 0.000 & --\\[1pt]
e-process & 0.0034 & --/\textbf{1.000} & \textbf{1.000} & 0.222 & 0.2635 & --/1.000 & 1.000 & 0.222 & 0.2767 & --/1.000 & 1.000 & 0.222\\[1pt]
Raw-p & 0.1943 & --/1.000 & 0.999 & 0.647 & 0.2687 & --/1.000 & 0.999 & 0.647 & 0.3753 & --/1.000 & 0.999 & 0.647\\[1pt]
\addlinespace[3pt]
\multicolumn{13}{c}{\textit{Alternating}}\\[2pt]
\textbf{FissionGuard} & 0.0394 & \textbf{0.984}/-- & 0.965 & 0.707 & 0.0004 & \textbf{0.723}/-- & \textbf{0.803} & 0.871 & 0.0000 & \textbf{0.781}/-- & \textbf{0.622} & \textbf{0.998}\\[1pt]
Splitting & 0.0367 & 0.980/-- & 0.969 & \textbf{0.711} & 0.0000 & 0.522/-- & 0.686 & \textbf{0.985} & 0.0000 & 0.000/-- & 0.000 & --\\[1pt]
e-process & 0.0016 & --/\textbf{0.988} & \textbf{0.998} & 0.222 & 0.1556 & --/1.000 & 0.998 & 0.222 & 0.4974 & --/1.000 & 0.998 & 0.222\\[1pt]
Raw-p & 0.1736 & --/1.000 & 0.998 & 0.570 & 0.2170 & --/1.000 & 0.998 & 0.570 & 0.6075 & --/1.000 & 0.998 & 0.570\\[1pt]
\bottomrule\end{tabular*}
\end{table}

At $\eta=0.25$, Table~\ref{tab:synthetic} shows that \method{} has higher power and recovers more watermarked tokens than Splitting, with the largest coverage advantages on Random and Alternating and comparable coverage on Contiguous. Both methods have low empirical FDR in these runs, while purity is similar on most layouts. Since the layouts have the same number of watermarked tokens and the same evidence law, these differences show how the arrangement of watermarked positions affects recovery.

At $\eta=0$, \method{} and Splitting recover similar amounts of watermarked content. Their differences become more pronounced when the reports must satisfy a positive proportion threshold. Splitting must accommodate the full-passage null budget using only its validation half, whereas \method{} retains validation values at every position. At $\eta=0.5$, \method{} retains substantial recovery on the block-based layouts but rarely reports on Random at $K=16$, while Splitting produces no reports under its full-passage validation budget. Stronger evidence at $K=64$ substantially improves recovery on Random at this threshold (Table~\ref{tab:synthetic-strength}). Both methods retain low empirical FDR across the tested thresholds. Appendix~\ref{app:synthetic-grid} examines shorter regions and different watermark proportions.

\method{} also achieves higher purity than both restart methods while keeping empirical FDR below the target. Both restart methods recover almost all watermarked tokens, but their FDR exceeds the target across all four layouts at $\eta=0.25$. Raw-p exceeds the FDR target even for presence, whereas e-process restart has low FDR there. The restart methods retain the same reports as $\eta$ increases, but more of those reports fail the stronger proportion claim. Fixed-interval experiments further support calibration of the partial-conjunction tests when the proportion null permits strong evidence at watermarked positions (Appendix~\ref{app:calibration}).

\subsection{Experiments on model-generated text}\label{sec:private-real}
We next compare the four layouts using Mistral-7B-Instruct-v0.3 and Qwen3-8B, with 100 original 4,096-token texts per model--layout group. Contiguous contains fully marked passages with varying lengths or gaps, whereas the other layouts contain four internally mixed regions, each exactly 75\% watermarked. Labels indicate whether the target watermark decoder was used at each position. Appendix Figure~\ref{fig:four-layouts} illustrates the layouts and the resulting reported and filtered passages. All methods are evaluated on the same texts. Generation settings and metric aggregation are described in Appendices~\ref{app:model-protocol} and~\ref{app:metric-conventions}, respectively.

\begin{table}[!htb]\centering
\caption{\textbf{Watermarked text across four layouts, $\alpha=0.1$, $\eta=0.25$.} Each model--layout group contains 100 original texts, at temperature 1.5 for Mistral and 2.0 for Qwen. Bold marks the highest displayed Power, ESR, coverage and purity among methods with FDR $\le\alpha$ (ties included; Power and ESR compared separately).}\label{tab:model-generated}
\small
\setlength{\tabcolsep}{2pt}
\renewcommand{\arraystretch}{1.06}
\begin{tabular*}{\linewidth}{@{\extracolsep{\fill}}l@{\hspace{4pt}}*{4}{r}@{\hspace{4pt}}*{4}{r}@{}}\toprule
 & \multicolumn{4}{c}{\textbf{Mistral-7B}} & \multicolumn{4}{c}{\textbf{Qwen3-8B}}\\
\cmidrule(lr){2-5}\cmidrule(lr){6-9}
Method & FDR & Power/ESR & Cov. & Purity & FDR & Power/ESR & Cov. & Purity\\\midrule
\multicolumn{9}{c}{\textit{Contiguous}}\\[2pt]
\textbf{FissionGuard} & 0.0040 & \textbf{0.939}/-- & \textbf{0.970} & 0.950 & 0.0065 & \textbf{0.971}/-- & \textbf{0.978} & 0.947\\[1pt]
Splitting & 0.0000 & 0.911/-- & 0.961 & \textbf{0.963} & 0.0000 & 0.940/-- & 0.966 & \textbf{0.966}\\[1pt]
e-process & 0.1564 & --/1.000 & 1.000 & 0.621 & 0.1572 & --/1.000 & 1.000 & 0.619\\[1pt]
Raw-p & 0.1760 & --/1.000 & 0.997 & 0.826 & 0.1883 & --/1.000 & 0.998 & 0.817\\[1pt]
\addlinespace[3pt]
\multicolumn{9}{c}{\textit{Random}}\\[2pt]
\textbf{FissionGuard} & 0.0057 & \textbf{0.948}/-- & \textbf{0.953} & 0.711 & 0.0035 & \textbf{0.959}/-- & \textbf{0.956} & 0.717\\[1pt]
Splitting & 0.0000 & 0.845/-- & 0.877 & \textbf{0.730} & 0.0000 & 0.897/-- & 0.911 & \textbf{0.728}\\[1pt]
e-process & 0.2713 & --/1.000 & 1.000 & 0.203 & 0.2711 & --/1.000 & 1.000 & 0.204\\[1pt]
Raw-p & 0.3568 & --/1.000 & 0.997 & 0.518 & 0.3677 & --/1.000 & 0.997 & 0.518\\[1pt]
\addlinespace[3pt]
\multicolumn{9}{c}{\textit{Gapped}}\\[2pt]
\textbf{FissionGuard} & 0.0017 & \textbf{0.888}/-- & \textbf{0.934} & 0.790 & 0.0039 & \textbf{0.937}/-- & \textbf{0.949} & 0.788\\[1pt]
Splitting & 0.0000 & 0.743/-- & 0.850 & \textbf{0.870} & 0.0000 & 0.782/-- & 0.861 & \textbf{0.877}\\[1pt]
e-process & 0.2714 & --/1.000 & 1.000 & 0.204 & 0.2698 & --/1.000 & 1.000 & 0.204\\[1pt]
Raw-p & 0.3450 & --/1.000 & 0.999 & 0.561 & 0.3605 & --/1.000 & 1.000 & 0.557\\[1pt]
\addlinespace[3pt]
\multicolumn{9}{c}{\textit{Alternating}}\\[2pt]
\textbf{FissionGuard} & 0.0068 & \textbf{0.941}/-- & \textbf{0.956} & 0.707 & 0.0039 & \textbf{0.951}/-- & \textbf{0.958} & 0.716\\[1pt]
Splitting & 0.0000 & 0.898/-- & 0.938 & \textbf{0.720} & 0.0000 & 0.937/-- & 0.953 & \textbf{0.722}\\[1pt]
e-process & 0.2700 & --/1.000 & 1.000 & 0.204 & 0.2704 & --/1.000 & 1.000 & 0.203\\[1pt]
Raw-p & 0.3504 & --/1.000 & 0.999 & 0.525 & 0.3704 & --/1.000 & 0.999 & 0.506\\[1pt]
\bottomrule\end{tabular*}
\end{table}
Table~\ref{tab:model-generated} shows that \method{} has higher power and token coverage than Splitting across both models and all four layouts, while maintaining empirical FDR well below the target. The observed coverage gains are larger on Random and Gapped, while both methods already recover most watermarked tokens on Contiguous and Alternating. The gains persist in qualified coverage (Table~\ref{tab:qualified-coverage}), which is close to total coverage for \method{} across all eight model--layout groups. Thus, most recovered watermarked tokens are covered by reports whose watermark proportions exceed $\eta$. Purity is also close to that of Splitting in most settings and substantially higher than that of both restart methods.

In contrast, the restart methods' near complete recovery on the Random, Gapped and Alternating settings comes with FDR above the target and substantially lower purity. Their reported passages therefore include a much larger proportion of unwatermarked tokens. The gains of \method{} in power and token coverage are obtained with the same alarm, localization and testing settings across both models and all four layouts. Additional controls and temperature comparisons are in Appendix~\ref{app:real-extension}.

\section{Discussion}

\method{} tests watermark proportions in candidate passages chosen as tokens arrive. Data fission allows every token to contribute to alarms, localization and validation, while partial-conjunction testing accounts for the watermarked positions allowed under the null. Under the stated conditions on the watermark pseudorandom variables and fission noise, the procedure controls FDR among reported passages for a fixed proportion threshold at every fixed finite horizon. We also derive a finite-sample lower bound on power among nonnull candidates.

Experiments on synthetic streams and text generated by Mistral and Qwen show higher power and token coverage than Splitting at the default proportion threshold, with similar purity in most settings and empirical FDR below the target. Further work can study the fission allocation together with the alarm and localization rules.

\label{end:main}

\bibliographystyle{plainnat}
\bibliography{referencesv2}
\clearpage
\appendix

\section{Validity of tests for selected passages}\label{app:proofs}
This appendix supplies the argument for the joint distribution of the token-level $p$-values and the proofs behind Section~\ref{sec:method}. The joint distribution of the token-level p-values supports validation after passage selection even when a passage satisfying the proportion null contains watermarked coordinates.

\subsection{Proof of Proposition~\ref{thm:factorization}}\label{app:watermark-input}

Several components of Assumption~\ref{ass:interface} are also used in existing watermark models. \citet[Working Hypotheses~2.1--2.2]{li2025statistical} assume independent pseudorandom variables, independence of the current variable from preceding text, and conditional independence between a human-generated token and its current pseudorandom variable given the joint past. In mixed-source text, \citet[Assumption~4.1(a)--(b)]{blanchet2026localization} likewise assume fresh pseudorandomness and conditional independence between the token and the current pseudorandom variable when watermark dependence is absent. At the pivot level, \citet[Assumption~1(i)]{su2026watermark} assume that each null pivot is uniform and independent of the preceding information.

The following generation model is sufficient for Assumption~\ref{ass:interface}. It is stronger than Assumption~\ref{ass:interface} and is also used for the Gumbel-max power analysis.

\begin{assumption}[Generation with independent pseudorandom variables]\label{ass:watermark-generation}
The variables $\zeta_t$ are mutually independent across token positions and requests, have the distributions specified by the watermarking scheme, and are independent of all non-key generation randomness. At each position, token generation depends only on the preceding text, non-key generation randomness and, at watermarked positions, the current $\zeta_t$. The watermark randomness at unwatermarked positions is not used by current or later generation and is not reused at other positions. Generation does not use the monitor's values or decisions.
\end{assumption}

To verify Assumption~\ref{ass:interface} under this model, condition on all non-key generation randomness and $(\zeta_t)_{t\in\mathcal M}$. The complete text is then fixed recursively. The remaining pseudorandom variables $(\zeta_t)_{t\in\mathcal N}$ retain their independent scheme-specified distributions, which do not depend on the quantities held fixed. They are therefore independent of the complete text and the watermarked pseudorandom variables, proving~\eqref{eq:unused-randomness}. The proof below uses only Assumption~\ref{ass:interface}.

\begin{proof}

We first show that Assumption~\ref{ass:interface}, together with the null distribution for each fixed token in Section~\ref{sec:token-evidence}, implies
\begin{equation}
(P_t)_{t\in\mathcal N}\overset{\mathrm{iid}}{\sim}\Unif(0,1),
\qquad
P_{\mathcal N}\perp P_{\mathcal M}.
\label{eq:unwatermarked-input-law}
\end{equation}
Throughout this subsection, watermark locations and request boundaries are fixed.

Define
\[
\mathcal G=\sigma\!\left(W_{1:\infty},(\zeta_t)_{t\in\mathcal M}\right).
\]
Assumption~\ref{ass:interface} gives $(\zeta_t)_{t\in\mathcal N}\perp\mathcal G$. These pseudorandom variables remain mutually independent with their scheme-specified distributions conditional on $\mathcal G$. The complete text and the token-level $p$-values $P_{\mathcal M}$ at watermarked positions are $\mathcal G$-measurable.

For each $t\in\mathcal N$, conditioning on $\mathcal G$ fixes $W_t$ but leaves $\zeta_t$ with its original distribution. The null distribution for each fixed token in Section~\ref{sec:token-evidence} thus gives
\[
\PP(Y_t\le y\mid\mathcal G)=F_0(y),
\qquad
\PP(P_t\le u\mid\mathcal G)=u,
\qquad 0\le u\le1.
\]
The second equality follows from the scheme-specific probability integral transform. For Gumbel-max, $P_t=1-U_{t,W_t}$, whereas for inverse-transform watermarking, $P_t=F_0(Y_t)$ uses the continuous null distribution given in Appendix~\ref{app:power}. In either case, the conditional distribution is the same for every fixed token.

Conditional on $\mathcal G$, each unwatermarked $P_t$ depends only on its own $\zeta_t$ and the fixed token $W_t$. Hence, for every finite $I\subseteq\mathcal N$ and deterministic thresholds $(u_t)_{t\in I}\in[0,1]^I$,
\begin{equation}
\PP\!\left(\bigcap_{t\in I}\{P_t\le u_t\}\,\middle|\,\mathcal G\right)
=
\prod_{t\in I}\PP(P_t\le u_t\mid\mathcal G)
=
\prod_{t\in I}u_t
\quad\text{almost surely}.
\label{eq:conditional-unwatermarked-law}
\end{equation}

The right-hand side does not depend on $\mathcal G$. For any event $A\in\mathcal G$, taking expectations therefore yields
\[
\PP\!\left(A\cap\bigcap_{t\in I}\{P_t\le u_t\}\right)
=\E\!\left[\1_A\prod_{t\in I}u_t\right]
=\PP(A)\prod_{t\in I}u_t.
\]
Taking $A$ to be the whole sample space establishes the independent uniform law for every finite subvector of $P_{\mathcal N}$. Allowing arbitrary $A\in\mathcal G$ establishes its independence from $\mathcal G$. Since finite-coordinate events generate the sigma-field of the sequence and $\sigma(P_{\mathcal M})\subseteq\mathcal G$, this proves \eqref{eq:unwatermarked-input-law}. Text tokens may remain dependent, as may the token-level $p$-values at watermarked positions.

For $t\in\mathcal N$, $\Phi^{-1}(P_t)$ and $V_t$ are independent standard normals. The matrix taking them to the two normal scores in~\eqref{eq:fission} is orthogonal. The scores are therefore independent standard normals, and their probability transforms are independent uniforms. Equation~\eqref{eq:unwatermarked-input-law} and independence of all fission noises imply that these unwatermarked pairs are mutually independent and independent of the complete sequence of alarm and validation pairs at watermarked positions. Hence all unwatermarked validation values are independent of all alarm values and of watermarked validation values. Conditioning on these sequences preserves the independent uniform distribution of the unwatermarked validation values. The statement for infinite sequences follows from its finite-dimensional versions by a monotone-class argument.
\end{proof}

The proof of Proposition~\ref{thm:factorization} uses the joint law in~\eqref{eq:unwatermarked-input-law}. Marginal super-uniformity alone does not provide this joint law.

\subsection{Online error control}
\begin{proof}[Proof of Theorem~\ref{cor:online-fdr}]
Fix the external design, complete alarm sequence, and all watermarked validation values. Alarm times, selected starts, candidate passages, truth labels, and the number of candidates by a fixed token horizon are then fixed. A true $H_{0j}^\eta$ contains at least $k_j=|B_j|-b_j$ unwatermarked coordinates. Select the first $k_j$ such coordinates in the deterministic token order and call this subset $I_0$. Under Proposition~\ref{thm:factorization}, $P_{I_0}^{\sf val}$ is iid uniform under this conditioning. Consequently, for $0\le u\le1$,
\[
\PP(Q_j^\eta\le u\mid P^{\sf alm},P_{\mathcal M}^{\sf val})
\le\PP\{\mathcal Q_{I_0}(P_{I_0}^{\sf val})\le u\mid P^{\sf alm},P_{\mathcal M}^{\sf val}\}\le u.
\]
The first inequality uses the maximum over all size-$k_j$ subsets. It does not assume the retained largest values themselves are iid uniform.

All watermarked validation values are constants after conditioning. Each $Q_j^\eta$ is thus a measurable function of unwatermarked validation coordinates on its own passage and fixed constants. Candidate passages are disjoint because the detector restarts after either testing outcome and each candidate passage uses only tokens observed since the most recent restart. These functions are mutually independent, including those corresponding to false nulls. Conditioning additionally on prior validation $p$-values preserves calibration of the selected tests under their null hypotheses, proving~\eqref{eq:selected-calibration}.

Apply the FDR theorem for independent $p$-values of the chosen online procedure to this fixed conditional sequence, with its fixed null labels and fixed candidate horizon. For SAFFRON, \citet[Theorem~1]{ramdas2018saffron} give the required FDR result under independence and monotonicity, whereas conditional super-uniformity alone yields a different mFDR result. The monotone LORD~1 and LORD~2 rules satisfy the FDR theorem for independent $p$-values of \citet[Theorem~3.1]{javanmard2018online}, and LORD++ satisfies the monotone generalized alpha-investing++ guarantee of \citet[Theorem~1]{ramdas2017decaying}. These statements use the prescribed initial wealth and reward constraints. The conditional expected FDP is at most $\alpha$. Integrating over the conditioned values and external design proves~\eqref{eq:alarm-fdr}. A zero-candidate sequence contributes zero.
\end{proof}

The guarantee allows arbitrary watermark strength, including within a true proportion null, and applies separately to each prespecified $\eta$ and fixed finite horizon. Each $\eta$ uses its own online testing history. Selecting the threshold or final horizon from the data, or merging reports from many thresholds, requires a separate analysis.

\section{Power bounds: proofs and watermark models}\label{app:power}\label{app:decoder-power}
This appendix gives the complete bounds underlying Section~\ref{sec:power}. We first control failure to reject on a deterministic interval, then account for selection and derive the Gumbel-max moment bounds. Appendix~\ref{app:theory-diagnostic} evaluates the finite-sample bound, and Appendix~\ref{app:inverse-power} treats inverse-transform evidence.

\subsection{Trimmed sums and the counting argument}\label{app:power-main-proof}
Throughout this appendix, probabilities and expectations condition on the fixed marking and request design. All unqualified sums over candidate indices $j$ are restricted to $\tau_j\le T$, and $z_u=\Phi^{-1}(1-u)$ for $0<u<1$. For a deterministic nonempty interval $B$, let $Q_B^\eta$ denote the $p$-value for the partial-conjunction test based on Stouffer's method, computed from the validation coordinates whether or not $B$ is selected.

\begin{lemma}[A bound for trimmed sums of normal scores]\label{lem:confirmation-tail}
For a deterministic nonempty interval $B$, put $n=|B|$, $b=\lfloor\eta n\rfloor$, and $Z_t=-\Phi^{-1}(P_t^{\sf val})$. For $0<u<1$,
\begin{equation}
\PP\{Q_B^\eta>u\}\le
\inf_{c\in\mathbb R,\,\theta\ge0}
\exp\!\left\{\theta\bigl(\sqrt{n-b}\,\Phi^{-1}(1-u)+bc\bigr)\right\}
\E\left[\exp\!\left\{-\theta\sum_{t\in B}\min(Z_t,c)\right\}\right].
\label{eq:confirmation-tail}
\end{equation}
\end{lemma}

\begin{proof}
Write $k=n-b$. Order the validation scores as $Z_{(1)}\le\cdots\le Z_{(n)}$ and put $S_B^\eta=\sum_{i=1}^k Z_{(i)}$. The maximum over retained subsets in~\eqref{eq:confirmation-interface} retains the $k$ smallest scores, so
\[
Q_B^\eta=\Phi(-S_B^\eta/\sqrt{k}).
\]
For every real score vector,
\begin{equation}
S_B^\eta=\sup_{c\in\mathbb R}
\left\{\sum_{t\in B}\min(Z_t,c)-bc\right\}.
\label{eq:trimmed-variational}
\end{equation}
Indeed, for each $c\in\mathbb R$,
\[
\sum_{t\in B}\min(Z_t,c)-bc =\sum_{i=1}^{k}\min(Z_{(i)},c)
 +\sum_{i=k+1}^{n}\{\min(Z_{(i)},c)-c\}\le\sum_{i=1}^{k}Z_{(i)}.
\]
If $b>0$, equality holds for $c\in[Z_{(k)},Z_{(k+1)}]$. If $b=0$, equality holds for every $c\ge Z_{(n)}$. These statements also cover ties, proving~\eqref{eq:trimmed-variational}.

For every deterministic $c$, the variational representation implies
\[
\{Q_B^\eta>u\}
=\{S_B^\eta<\sqrt{k}z_u\}
\subseteq
\left\{\sum_{t\in B}\min(Z_t,c)<\sqrt{k}z_u+bc\right\}.
\]
Applying Markov's inequality to the exponential of the negative sum gives, for $\theta>0$,
\[
\PP\{Q_B^\eta>u\}
\le
e^{\theta(\sqrt{k}z_u+bc)}
\E\left[\exp\!\left\{-\theta\sum_{t\in B}\min(Z_t,c)\right\}\right].
\]
The same inequality holds for $\theta=0$, when its right side is one. Taking the infimum over deterministic $c$ and $\theta$ proves~\eqref{eq:confirmation-tail}. Expectations may be interpreted in the extended sense, but the choice $\theta=0$ is always finite.

To evaluate the expectation, let $\mathcal F_{t-1}$ contain the preceding generation and fission history. Suppose deterministic finite bounds $M_t(c,\theta)>0$ satisfy
\[
\E[e^{-\theta\min(Z_t,c)}\mid\mathcal F_{t-1}]
\le M_t(c,\theta),
\qquad c\in\mathbb R,\quad\theta\ge0,
\qquad M_t(c,0)=1.
\]
The exponential factors from all earlier positions are $\mathcal F_{t-1}$-measurable. Starting with the final position of $B$ and applying the tower property successively therefore yields
\[
\E\left[\exp\!\left\{-\theta\sum_{t\in B}\min(Z_t,c)\right\}\right]
\le\prod_{t\in B}M_t(c,\theta).
\]
Consequently,
\begin{equation}
\PP\{Q_B^\eta>u\}\le\varepsilon_B(u):=
\inf_{c\in\mathbb R,\,\theta\ge0}
\exp\!\left\{\theta(\sqrt{k}z_u+bc)
 +\sum_{t\in B}\log M_t(c,\theta)\right\}.
\label{eq:conditional-cap-tail}
\end{equation}
This argument allows dependence among watermarked scores through the preceding history. Each deterministic pair $(c,\theta)$ bounds the same failure event, so optimizing these analysis parameters preserves the bound. Appendix~\ref{app:gumbel-power} derives suitable conditional moment bounds from the Gumbel-max NTP probabilities.
\end{proof}

\begin{proof}[Proof of Theorem~\ref{thm:main-power}]
For this proof and the subsequent discussion, write
\[
\Gamma_\eta(T)=\E\left[\sum_jD_j^\eta\right]>0,
\qquad
L_\eta(T)=\E\left[\sum_jD_j^\eta\1\{B_j\notin\mathcal B_T\}\right].
\]
These are the expected number of nonnull candidates through $T$ and the expected number of such candidates outside $\mathcal B_T$, respectively. A nonnull candidate is not reported only if $Q_j^\eta>\alpha_j\ge u$. Moreover, disjoint nonempty candidate passages ensure that each deterministic interval is selected at most once. Thus, on every realization,
\[
\sum_jD_j^\eta(1-R_j)
\le
\sum_jD_j^\eta\1\{B_j\notin\mathcal B_T\}
+\sum_{B\in\mathcal B_T}\1\{Q_B^\eta>u\}.
\]
For a candidate inside $\mathcal B_T$, a failure contributes to the corresponding deterministic-interval event on the right. For a candidate outside $\mathcal B_T$, its contribution is bounded by the first sum. The value $Q_B^\eta$ is defined from the validation coordinates for every deterministic $B$, including those not selected on that realization.

Taking expectations gives
\[
\E\left[\sum_jD_j^\eta(1-R_j)\right]
\le L_\eta(T)+\sum_{B\in\mathcal B_T}\PP\{Q_B^\eta>u\}.
\]
Since
\[
\mathrm{Pow}_\eta(T)
=1-\frac{\E\left[\sum_jD_j^\eta(1-R_j)\right]}{\Gamma_\eta(T)}
\quad\text{and}\quad \mathrm{Pow}_\eta(T)\ge0,
\]
division by $\Gamma_\eta(T)>0$ proves~\eqref{eq:main-power}. Substituting~\eqref{eq:conditional-cap-tail} also gives the computable form
\[
\mathrm{Pow}_\eta(T)\ge
\left[1-\frac{L_\eta(T)+\sum_{B\in\mathcal B_T}\varepsilon_B(u)}
{\Gamma_\eta(T)}\right]_+,
\]
which is evaluated in Appendix~\ref{app:theory-diagnostic}.
\end{proof}

\paragraph{SAFFRON testing levels and the length requirement.}
In~\eqref{eq:saffron}, $1\le j-C_{1:j-1}\le j$, the sequence $(\gamma_j)$ is nonincreasing, and all reward terms are nonnegative. Thus
\begin{equation}
\alpha_j\ge\min\{\lambda,(1-\lambda)w_0\gamma_j\}.
\label{eq:level-floor}
\end{equation}
A deterministic candidate bound $J_{\max}$ gives the floor $u=\min\{\lambda,(1-\lambda)w_0\gamma_{J_{\max}}\}$. Disjoint nonempty candidates imply $J_{\max}\le T$. With the experimental schedule $\gamma_j=j^{-1.6}/\zeta(1.6)$, we may take
\[
u_T=\min\{\lambda,(1-\lambda)w_0\gamma_T\},
\qquad \log(1/u_T)=1.6\log T+O(1).
\]
Thus testing levels have a polynomial lower bound, and $n_{\min}\ge C'\log(eT)$ suffices for the length requirement in Corollary~\ref{cor:exponential-power} when $C'$ is sufficiently large. The proof of the corollary below gives explicit sufficient conditions for power to tend to one. The experimental selector has a fixed maximum length of 256 tokens, so these experiments concern finite horizons rather than asymptotic consistency. Appendix~\ref{app:theory-diagnostic} bounds $\Gamma_\eta(T)$ from the model, while observed counts are used only for empirical power.

\subsection{Gumbel-max: token laws and exponential bounds}\label{app:gumbel-power}
We first derive the watermarked token law and a deterministic bound for its conditional capped moments, then prove Corollary~\ref{cor:exponential-power}. Watermark locations and request boundaries are fixed, and generation satisfies Assumption~\ref{ass:watermark-generation}. The NTP probabilities are predictable with respect to the preceding generation history, and the current pseudorandom variable $\zeta_t$ and fission noise are independent of the past. Write $\phi$ for the standard normal density.

For a probability vector $\pi$, let $(U_v)_{v\in\mathcal V}$ be independent uniform variables. The Gumbel-max decoder selects
\[
W=\underset{v:\pi_v>0}{\arg\max}\,\frac{\log U_v}{\pi_v},
\qquad Y=U_W,\qquad P=1-Y.
\]
The variables $-\log U_v/\pi_v$ are independent exponentials with rates $\pi_v$. Since these rates sum to one, their minimum $E$ has the $\operatorname{Exp}(1)$ distribution and is independent of the winning token $W$, whose distribution is $\pi$. Moreover, $U_W=e^{-\pi_WE}$. It follows that, for $0<p<1$,
\begin{equation}
\PP(P\le p)=1-\sum_{v:\pi_v>0}\pi_v(1-p)^{1/\pi_v},\qquad
f_\pi^G(p)=\sum_{v:\pi_v>0}(1-p)^{1/\pi_v-1}.
\label{eq:gumbel-density}
\end{equation}
The watermarked distribution of the Gumbel-max pivotal statistic \citep{li2025statistical} gives this $p$-value law. In particular, a uniform probability vector on $K$ tokens gives $P\sim\operatorname{Beta}(1,K)$.

Suppose $\max_v\pi_v\le1-\delta$ for some $0<\delta<1$. Put $a_\delta=1/(1-\delta)$ and let $P_\delta$ have density
\[
g_\delta(p)=a_\delta(1-p)^{a_\delta-1},\qquad 0<p<1.
\]
Since $1/\pi_v\ge a_\delta$ whenever $\pi_v>0$, equation~\eqref{eq:gumbel-density} gives
\[
\PP(P\le p)\ge1-(1-p)^{a_\delta}=\PP(P_\delta\le p).
\]
Thus $P$ is stochastically no larger than $P_\delta$, and $P_\delta$ is stochastically no larger than a uniform variable. This comparison holds conditionally on the past whenever the predictable vector $\pi_t$ satisfies the same cap. The reference variable $P_\delta$ is defined for every $\delta\in(0,1)$. A probability vector on a vocabulary of size $d$ can satisfy the cap only when $d(1-\delta)\ge1$.

Let $G\sim N(0,1)$ be independent of $P_\delta$, and define
\[
Z_\delta=-\sqrt{1-\rho}\,\Phi^{-1}(P_\delta)+\sqrt\rho\,G.
\]
Independent Gaussian fission noise preserves stochastic order. Consequently, the conditional distribution of a watermarked validation score $Z_t=-\Phi^{-1}(P_t^{\sf val})$ stochastically dominates the distribution of $Z_\delta$, which in turn stochastically dominates $N(0,1)$. For $c\in\mathbb R$ and $\theta\ge0$, the function $z\mapsto e^{-\theta\min(z,c)}$ is nonincreasing. The conditional moment bounds in~\eqref{eq:conditional-cap-tail} may therefore be taken as
\begin{equation}
M_0(c,\theta)=\E\bigl[ e^{-\theta\min(G,c)}\bigr],\qquad
M_\delta(c,\theta)=\E\bigl[ e^{-\theta\min(Z_\delta,c)}\bigr]\le M_0(c,\theta),
\label{eq:cap-envelopes}
\end{equation}
at unwatermarked and watermarked positions, respectively. The unwatermarked moment is exact. Splitting its Gaussian integral at $c$ gives
\begin{equation}
M_0(c,\theta)=e^{\theta^2/2}\Phi(c+\theta)+e^{-\theta c}\Phi(-c).
\label{eq:normal-cap-moment}
\end{equation}
For the watermarked moment, write $a_\rho(p)=-\sqrt{1-\rho}\,\Phi^{-1}(p)$. The corresponding calculation for a $N(a,\rho)$ variable yields
\begin{align}
H_\rho(c,\theta;a)
&=e^{-\theta a+\rho\theta^2/2}\Phi\!\left(\frac{c-a+\rho\theta}{\sqrt\rho}\right)
  +e^{-\theta c}\Phi\!\left(\frac{a-c}{\sqrt\rho}\right),
\label{eq:cap-kernel}\\
M_\delta(c,\theta)
&=\int_0^1g_\delta(p)H_\rho(c,\theta;a_\rho(p))\,dp.
\label{eq:gumbel-cap-moment}
\end{align}
For the uniform-$K$ decoder, $\delta=1-1/K$ makes $g_\delta$ the actual watermarked density.

Let a deterministic interval $B$ contain $n$ positions, of which $m$ are watermarked, and put $b=\lfloor\eta n\rfloor$ and $k=n-b$. Successive conditional expectations, followed by~\eqref{eq:conditional-cap-tail}, give
\begin{equation}
\begin{aligned}
\PP\{Q_B^\eta>u\}&\le\varepsilon_B(u)\\
&:=\inf_{c\in\mathbb R,\,\theta\ge0}
\exp\!\left\{\theta(\sqrt{k}z_u+bc)
 +(n-m)\log M_0(c,\theta)+m\log M_\delta(c,\theta)\right\}.
\end{aligned}
\label{eq:homogeneous-cap-tail}
\end{equation}
This bound applies to predictable NTP probabilities through the conditional moment bounds above.

The moments are finite and differentiable in a neighborhood of $\theta=0$. Indeed, $g_\delta\le a_\delta$, so $-\Phi^{-1}(P_\delta)$ has density bounded by $a_\delta\phi$. Scaling and convolution with the independent Gaussian noise show that $Z_\delta$ also has density bounded by $a_\delta\phi$. Thus $\E\bigl[ e^{q|Z_\delta|}\bigr]<\infty$ for every $q>0$, which justifies differentiation under the capped-moment expectations. In particular,
\[
M_0(c,0)=M_\delta(c,0)=1,
\]
and
\[
\left.\frac{\partial}{\partial\theta}\log M_0(c,\theta)\right|_{\theta=0}
=-\E[\min(G,c)],\qquad
\left.\frac{\partial}{\partial\theta}\log M_\delta(c,\theta)\right|_{\theta=0}
=-\E[\min(Z_\delta,c)].
\]

\phantomsection\label{app:power-statements}
\paragraph{Conditions for Corollary~\ref{cor:exponential-power}.}
In the setting of Theorem~\ref{thm:main-power}, use the partial-conjunction test based on Stouffer's method in~\eqref{eq:pc-stouffer}. Assume Gumbel-max generation satisfying Assumption~\ref{ass:watermark-generation}, with predictable NTP probabilities and independent fission noise.
Fix $0\le\eta<\kappa\le1$ and $0<\rho<1$. There exists $\delta_*\in[0,1)$, depending only on $(\rho,\kappa,\eta)$, with the following property. For each fixed $\delta\in(\delta_*,1)$ such that $\max_v\pi_{t,v}\le1-\delta$ almost surely at every watermarked position, there are constants $C,c_1>0$ depending only on $(\delta,\rho,\kappa,\eta)$. If every $B\in\mathcal B_T$ has length at least $n_{\min}$ and watermark proportion at least $\kappa$, and $n_{\min}\ge C\{1+\log(1/u)\}$, then
\begin{equation}
\sum_{B\in\mathcal B_T}\PP\{Q_B^\eta>u\}
\le |\mathcal B_T|e^{-c_1n_{\min}}.
\label{eq:exponential-power}
\end{equation}
Substituting this bound into~\eqref{eq:main-power} gives~\eqref{eq:exponential-power-bound}. For $\eta=0$, one can take $\delta_*=0$.
For $\eta>0$, sufficient signal must remain after trimming. Appendix~\ref{app:gumbel-power} gives the criterion and a weak-signal example. A vocabulary of size $d$ requires $\delta\le1-1/d$.

\begin{proof}[Proof of Corollary~\ref{cor:exponential-power}]
For fixed $0\le\eta<\kappa\le1$ and $0<\rho<1$, define
\begin{equation}
r(c)=(1-\kappa)\E[\min(G,c)]+\kappa\E[\min(Z_\delta,c)]-\eta c.
\label{eq:trimmed-signal}
\end{equation}
We first show that $r(c)>0$ for some deterministic $c>0$ whenever the probability cap is sufficiently strong. For every fixed $c>0$,
\[
\E[\min(G,c)]=c\Phi(-c)-\phi(c).
\]
Also, if $A_\delta=-\Phi^{-1}(P_\delta)$, then
\[
\PP(A_\delta\le x)=\Phi(x)^{1/(1-\delta)}.
\]
Thus $A_\delta\to+\infty$ in probability as $\delta\uparrow1$, and the same holds for $Z_\delta$. Since $Z_\delta$ stochastically dominates a standard normal variable,
\[
0\le c-\E[\min(Z_\delta,c)]
=\int_{-\infty}^c\PP(Z_\delta\le x)\,dx,
\qquad
\PP(Z_\delta\le x)\le\Phi(x).
\]
The function $\Phi$ is integrable on $(-\infty,c]$, so dominated convergence gives $\E[\min(Z_\delta,c)]\to c$. Consequently,
\[
r(c)\longrightarrow(\kappa-\eta)c+(1-\kappa)\E[\min(G,c)].
\]
For $c>0$, $\E[\min(G,c)]\ge-\phi(0)$. Taking
\[
c=1+\frac{(1-\kappa)\phi(0)}{\kappa-\eta}
\]
makes the last limit at least $\kappa-\eta>0$. The reference scores increase stochastically with $\delta$. It follows that there exists $\delta_*\in[0,1)$, depending only on $(\rho,\kappa,\eta)$, such that this choice of $c$ has $r(c)>0$ for every $\delta>\delta_*$. For a vocabulary of size $d$, the applicable caps satisfy $\delta\in(\delta_*,1-1/d]$. This feasible range is nonempty precisely when $\delta_*<1-1/d$.

When $\eta=0$, each fixed $\delta>0$ suffices. The distribution of $Z_\delta$ strictly stochastically dominates $N(0,1)$ and has a finite first moment, so $\E[ Z_\delta]>0$. Letting $c\to\infty$ in~\eqref{eq:trimmed-signal} gives
\[
r(c)\longrightarrow\kappa\E[ Z_\delta]>0.
\]
We may therefore take $\delta_*=0$ in this case and choose $c>0$ after fixing $\delta$.

Now fix an admissible $\delta$ and a deterministic $c>0$ for which $r(c)>0$. Put
\[
\Lambda_c(\theta)=\theta\eta c
 +(1-\kappa)\log M_0(c,\theta)+\kappa\log M_\delta(c,\theta).
\]
The moment derivatives above imply $\Lambda_c(0)=0$ and $\Lambda_c'(0)=-r(c)<0$. Hence there exist $\theta_0>0$ and $c_0>0$ such that $\Lambda_c(\theta_0)=-c_0$. If $m/n\ge\kappa$, then $M_\delta\le M_0$ and $bc\le\eta nc$ give
\begin{align*}
&\theta_0(\sqrt{k}z_u+bc)
 +(n-m)\log M_0(c,\theta_0)+m\log M_\delta(c,\theta_0)\\
&\hspace{2em}\le\theta_0\sqrt{k}z_u+n\Lambda_c(\theta_0)
\le\theta_0\sqrt n\,[z_u]_+-c_0n.
\end{align*}
The last expression is at most $-c_0n/2$ whenever
\[
n\ge\left(\frac{2\theta_0[z_u]_+}{c_0}\right)^2.
\]
Since $[z_u]_+^2\le2\log(1/u)$, it is enough to take
\[
C=1+\frac{8\theta_0^2}{c_0^2},\qquad c_1=\frac{c_0}{2},
\]
and require $n\ge C\{1+\log(1/u)\}$. Substitution into~\eqref{eq:homogeneous-cap-tail} proves
\begin{equation}
\PP\{Q_B^\eta>u\}\le e^{-c_1|B|}.
\label{eq:interval-exponential-tail}
\end{equation}
The choices of $C$ and $c_1$ depend only on $(\delta,\rho,\kappa,\eta)$. Summing~\eqref{eq:interval-exponential-tail} over $B\in\mathcal B_T$ and using $|B|\ge n_{\min}$ proves~\eqref{eq:exponential-power}. Combining this conclusion with Theorem~\ref{thm:main-power} gives the corresponding bound on power among nonnull candidates.

For the asymptotic assertion, consider configurations indexed by $T$, with $(\delta,\rho,\kappa,\eta)$ fixed and the minimum candidate length allowed to depend on $T$. Suppose $u_T\ge aT^{-\beta}$ and $\Gamma_\eta(T)\ge g_0$ for fixed $a,\beta,g_0>0$ and all sufficiently large $T$. Assume almost surely that every nonnull candidate through $T$ has watermark proportion at least $\kappa$ and length at least $n_{\min}(T)$, where
\[
n_{\min}(T)\ge A\log(eT),
\qquad A>\max\{\beta C,\,2/c_1\}.
\]
Take $\mathcal B_T$ to contain all intervals in $\{1,\ldots,T\}$ with length at least $n_{\min}(T)$ and watermark proportion at least $\kappa$. This family is deterministic under the fixed marking design, contains every nonnull candidate almost surely, and has at most $T(T+1)/2$ members. Hence $L_\eta(T)=0$. Since $\log(1/u_T)\le\beta\log T+\log(1/a)$ and $A>\beta C$, the length requirement holds for all sufficiently large $T$. Equation~\eqref{eq:exponential-power-bound} then gives
\[
0\le1-\mathrm{Pow}_\eta(T)
\le\frac{T(T+1)}{2g_0}e^{-c_1n_{\min}(T)}
=O(T^{2-c_1A})\longrightarrow0.
\]
For the SAFFRON schedule above, $\beta=1.6$, so any $A>\max\{1.6C,\,2/c_1\}$ suffices.
\end{proof}

\paragraph{A condition for positive signal after trimming.}
The proof applies at any fixed $\delta>0$ for which $r(c)>0$ for some $c>0$. This condition has a direct interpretation. Let $F_\kappa$ be the distribution function of the mixture
\[
(1-\kappa)\mathcal L(G)+\kappa\mathcal L(Z_\delta).
\]
For $\eta>0$, differentiation gives $r'(c)=1-F_\kappa(c)-\eta$. The maximum is therefore attained at the $(1-\eta)$ quantile. At this quantile, the contribution $c\{1-F_\kappa(c)-\eta\}$ is zero, and integration gives
\begin{equation}
\sup_{c\in\mathbb R}r(c)
=\int_0^{1-\eta}F_\kappa^{-1}(v)\,dv.
\label{eq:trimmed-signal-integral}
\end{equation}
For $\eta=0$, the same identity follows by letting $c\to\infty$. If $c\le0$, then $r(c)\le(1-\eta)c\le0$, so restricting attention to $c>0$ excludes no positive value. Thus~\eqref{eq:trimmed-signal} requires a positive mean contribution after removing the largest $\eta$ fraction of the reference scores. The stochastic comparison then supplies the same sufficient condition for the actual NTP probabilities.

\paragraph{Small proportion thresholds.}
For fixed $\delta,\rho\in(0,1)$ and $\kappa\in(0,1]$, choose $c>0$ sufficiently large that $s=(1-\kappa)\E[\min(G,c)]+\kappa\E[\min(Z_\delta,c)]>0$. Then every $0\le\eta<\min\{\kappa/2,s/(2c)\}$ gives $r(c)=s-\eta c\ge s/2>0$, so the exponential argument applies.

\paragraph{A weak-signal example.}
The inequalities $\kappa>\eta$ and $\delta>0$ alone do not ensure high power for the partial-conjunction test based on Stouffer's method when $\eta>0$. Fix $\rho\in(0,1)$ and $\eta\in(0,1)$, mark every position, and use the two-token Gumbel-max decoder $\pi=(1-\delta,\delta)$ with $0<\delta<1/2$ and independent watermark pseudorandom variables across token positions. The resulting watermarked $p$-value density is
\[
f_\delta(p)=(1-p)^{\delta/(1-\delta)}+(1-p)^{1/\delta-1}.
\]
As $\delta\downarrow0$, this density converges almost everywhere to one and is bounded by two. The validation-score distributions therefore converge to $N(0,1)$, and their densities are bounded by $2\phi$. This bound also gives uniformly integrable absolute moments. Their quantile functions consequently converge in $L^1(0,1)$, so their lower $(1-\eta)$ trimmed integrals converge to
\[
\int_0^{1-\eta}\Phi^{-1}(v)\,dv
=-\phi\{\Phi^{-1}(1-\eta)\}<0.
\]
Choose a sufficiently small fixed $\delta>0$ for which the corresponding trimmed integral is negative, and denote the resulting validation-score distribution function by $F_\delta^{\rm act}$. These validation scores are iid. For the prefix intervals $B_n=\{1,\ldots,n\}$, let $\widehat F_n$ be their empirical distribution function, put $b_n=\lfloor\eta n\rfloor$, and let $S_{B_n}^\eta$ be the sum of the $n-b_n$ smallest validation scores. The finite first moment implies almost-sure $L^1(0,1)$ convergence of the empirical quantile functions. Hence
\[
\frac{S_{B_n}^\eta}{n}
=\int_0^{1-b_n/n}\widehat F_n^{-1}(v)\,dv
\longrightarrow\int_0^{1-\eta}(F_\delta^{\rm act})^{-1}(v)\,dv<0
\qquad\text{almost surely}.
\]
At any fixed level $u\in(0,1)$, or for levels $0<u_n\le\bar u<1$ satisfying $\log(1/u_n)=o(n)$, the normalized rejection threshold obeys
\[
\frac{\sqrt{n-b_n}\,\Phi^{-1}(1-u_n)}{n}\longrightarrow0.
\]
The rejection probability of this partial-conjunction test therefore tends to zero in this model. The example concerns this test and does not rule out other tests of the proportion null.

\paragraph{Alternative bound over retained subsets.}
For the partial-conjunction test based on Stouffer's method, retain the notation $n=|B|$, $m=|B\cap\mathcal M|$, $b=\lfloor\eta n\rfloor$, $k=n-b$, and $z_u=\Phi^{-1}(1-u)$. An upper bound on the untrimmed watermarked log moment is:
\begin{equation}
\psi(\theta)=\frac{\rho\theta^2}{2}+\log\!\int_0^1g_\delta(p)e^{\theta\sqrt{1-\rho}\Phi^{-1}(p)}\,dp.
\label{eq:gumbel-mgf}
\end{equation}
Let $\psi_t(\theta)$ be $\theta^2/2$ at unwatermarked positions and $\psi(\theta)$ at watermarked positions. For each fixed retained subset $I$ of size $k$, exponential Markov and successive conditional expectations bound its nonrejection probability by $\exp\{\theta\sqrt{k}z_u+\sum_{t\in I}\psi_t(\theta)\}$. The event $Q_B^\eta>u$ is the union of these subset events. Using a common $\theta$ and grouping subsets by their retained watermarked count $a$ gives the separate valid bound
\begin{equation}
\begin{aligned}
\PP\{Q_B^\eta>u\}\le\varepsilon_B^{\rm UB}(u):=
\min\!\Big\{1,\inf_{\theta\ge0}&e^{\theta\sqrt{k}z_u}\\
{}\times&\sum_{a=\max(0,k-n+m)}^{\min(m,k)}\binom{m}{a}\binom{n-m}{k-a}
e^{a\psi(\theta)+(k-a)\theta^2/2}\Big\}.
\end{aligned}
\label{eq:homogeneous-tail}
\end{equation}
Optimizing each subset separately can only improve this union bound. With all positions watermarked, the common-parameter sum is $\binom{n}{b}e^{k\psi(\theta)}$. At $\eta=0$, there is one subset, yielding the presence Chernoff bound, and the capped bound recovers the same expression by letting $c\to\infty$. Either bound, or their pointwise minimum, may be used in Theorem~\ref{thm:main-power}. No uniform ordering of their optimized values is asserted. Appendix~\ref{app:theory-diagnostic} compares conservative numerical evaluations of both.

\subsection{Inverse transform: finite vocabulary}\label{app:inverse-power}
For the inverse decoder, let $\sigma$ be a uniform random token permutation, independent of $U\sim\Unif(0,1)$. For a vocabulary of size $d\ge2$, write $v_1,\ldots,v_d$ for the tokens in permutation order, put $c_{\sigma,k}=\sum_{l\le k}\pi_{v_l}$ with $c_{\sigma,0}=0$, and set $b_k=(k-1)/(d-1)$. The decoder chooses $W=v_k$ when $c_{\sigma,k-1}<U\le c_{\sigma,k}$. Its pivotal statistic is $Y=|U-b_{\sigma(W)}|$, where $\sigma(W)$ denotes the rank. An unwatermarked token has uniform rank independent of $U$. For $0\le y\le1$, the exact null and watermarked distributions are
\begin{align}
F_{0,d}(y)&=d^{-1}\sum_{k=1}^d[\min(1,b_k+y)-\max(0,b_k-y)],\label{eq:inverse-null}\\
F_{\pi,d}^Y(y)&=(d!)^{-1}\sum_\sigma\sum_{k=1}^d
 |(c_{\sigma,k-1},c_{\sigma,k}]\cap[b_k-y,b_k+y]|.\label{eq:inverse-alt}
\end{align}
These distribution functions follow by integrating over the values of $U$ for which the pivotal statistic is at most $y$, as in \citet[Lemma 4.1]{li2025statistical}. The null distribution is continuous and strictly increasing on $[0,1]$, so $P=F_{0,d}(Y)$ is uniform under the null. Using $U$ alone would discard the evidence from inverse-transform watermarking. The density of the watermarked $p$-value is
\[
f_{\pi,d}^I(p)=\frac{(F_{\pi,d}^Y)'(F_{0,d}^{-1}(p))}{F_{0,d}'(F_{0,d}^{-1}(p))}\quad\text{almost everywhere}.
\]
The null derivative is at least $2/d$, and the watermarked derivative is at most $2d$, yielding $0\le f_{\pi,d}^I\le d^2$. For any declared nonempty finite-vocabulary class $\mathcal C$, an upper bound on the capped moment is
\begin{equation}
M_{\mathcal C}^I(c,\theta)=\sup_{\pi\in\mathcal C}\int_0^1 f_{\pi,d}^I(p)
 H_\rho(c,\theta;a_\rho(p))\,dp.
\label{eq:inverse-cap-moment}
\end{equation}
The density bound and Gaussian exponential integrability make this finite, with $M_{\mathcal C}^I(c,0)=1$. A predictable $\pi_t\in\mathcal C$ therefore satisfies the conditional moment bound in Lemma~\ref{lem:confirmation-tail}. It uses the exact distribution functions in~\eqref{eq:inverse-null}--\eqref{eq:inverse-alt}. No finite-vocabulary stochastic ordering is assumed. The alternative untrimmed moment bound for~\eqref{eq:homogeneous-tail} is obtained from
\[
\psi_{\mathcal C}^I(\theta)=\rho\theta^2/2+\log\sup_{\pi\in\mathcal C}
\int_0^1f_{\pi,d}^I(p)e^{\theta\sqrt{1-\rho}\Phi^{-1}(p)}\,dp.
\]
Assume inverse-transform generation satisfying Assumption~\ref{ass:watermark-generation}, with predictable NTP probabilities and independent fission noise. Fix $d\ge2$, $\rho\in(0,1)$, $0\le\eta<\kappa\le1$, and a nonempty class $\mathcal C$, independently of $T$, such that $\pi_t\in\mathcal C$ almost surely at every watermarked position. Suppose deterministic constants $c,\theta_0,c_0>0$, independent of $T$, satisfy
\begin{equation}
\theta_0\eta c+(1-\kappa)\log M_0(c,\theta_0)+\kappa\log M_{\mathcal C}^I(c,\theta_0)\le-c_0.
\label{eq:inverse-exponential-condition}
\end{equation}
Since $M_0(c,\theta_0)\ge e^{\theta_0^2/2}>1$, this condition implies $M_{\mathcal C}^I(c,\theta_0)<M_0(c,\theta_0)$. Thus, for every deterministic passage $B$ of length $n$ and watermark proportion at least $\kappa$, equation~\eqref{eq:conditional-cap-tail} gives $\PP\{Q_B^\eta>u\}\le\exp\{\theta_0\sqrt n\,[z_u]_+-c_0n\}$. Taking $C=1+8\theta_0^2/c_0^2$ and $c_1=c_0/2$, we obtain $\PP\{Q_B^\eta>u\}\le e^{-c_1n}$ whenever $n\ge C\{1+\log(1/u)\}$. If every $B\in\mathcal B_T$ has watermark proportion at least $\kappa$ and length at least $n_{\min}\ge C\{1+\log(1/u)\}$, summing this bound and applying Theorem~\ref{thm:main-power} gives~\eqref{eq:exponential-power-bound}. If all nonnull candidates have watermark proportion at least $\kappa$, the same asymptotic conditions on candidate lengths, testing levels and expected nonnull candidate counts used in the proof of Corollary~\ref{cor:exponential-power} then yield $\mathrm{Pow}_\eta(T)\to1$.

\section{Implementation and comparison protocols}\label{app:implementation}
The comparisons use alarms based on Stouffer's method, with one fission allocation and a separate SAFFRON history for each deployment and proportion threshold. Comparator passages follow the reporting rules below.

\subsection{Alarm, localization, and testing history}
For a deterministic nonempty subset $I$, Stouffer's method \citep{stouffer1949american} gives the one-sided combined $p$-value
\begin{equation}
\mathcal Q_I((p_t)_{t\in I})=\Phi\!\left(|I|^{-1/2}\sum_{t\in I}\Phi^{-1}(p_t)\right).\label{eq:stouffer}
\end{equation}
Under the joint null that all positions in $I$ are unwatermarked, the input $p$-values are independent uniforms under the stated model. At a fixed level $u$, the corresponding one-sided Stouffer combination test rejects this joint null when $\mathcal Q_I\le u$.
With $\rho=1/2$, the alarm rule scans suffixes of 64, 128 and 256 raw tokens and uses the localization criterion below. All suffixes must use only tokens observed since the most recent restart. The alarm threshold of 0.01 selects candidates. It is distinct from the level $\alpha_j$ assigned to the final proportion test. Settings are fixed across layouts, signal strengths, $\eta$, models and temperatures within each comparison. The detector restarts after each alarm, regardless of the testing decision, and at each new request. An alarm is eligible only after at least 64 tokens have been processed since the most recent restart. This does not impose a 64-token minimum on the localized passage.

At an alarm, admissible starts give passage lengths of 32--256 raw tokens. A working density family $f_\mu$ can rank starts using its profile log likelihood relative to the uniform null. For the positive shift family of $-\Phi^{-1}(p)$, this reduces to
\begin{equation}
\ell_j(s)=\sup_{\mu\ge0}\sum_{t=s}^{\tau_j}\log\frac{f_\mu(P_t^{\sf alm})}{f_0(P_t^{\sf alm})}
=\frac{\left(-\sum_{t=s}^{\tau_j}\Phi^{-1}(P_t^{\sf alm})\right)_+^2}{2(\tau_j-s+1)}.\label{eq:profile-likelihood}
\end{equation}
The implementation maximizes the corresponding positive standardized sum and chooses the earliest start in a tie. An alarm supplies an admissible positive-sum suffix, so the maximizers coincide. 

For a common symmetric, coordinatewise nondecreasing combination rule \citep{heller2011discussion}, the maximum in~\eqref{eq:confirmation-interface} is attained by the $|B_j|-b_j$ largest validation values. Our implementation uses Stouffer's method. If $p_{(1)}\le\cdots\le p_{(n)}$ are the validation values on $B_j$, with $n=|B_j|$, then
\begin{equation}
Q_j^\eta=\Phi\!\left(\frac{\sum_{i=b_j+1}^{n}\Phi^{-1}(p_{(i)})}{\sqrt{n-b_j}}\right).
\label{eq:pc-stouffer}
\end{equation}
This substitutes~\eqref{eq:stouffer} into~\eqref{eq:confirmation-interface}. For a deterministic candidate interval $B$, use the same formula with $n=|B|$ and $b=\lfloor\eta n\rfloor$, and denote its value by $Q_B^\eta$.
For a fixed retained size, the combined $p$-value is increasing in each coordinate, so the $n-b$ largest validation values achieve the maximum over subsets. Equivalently, for the normal scores $z=\Phi^{-1}(1-p)$, this removes the $b$ largest scores and sums the remaining smallest ones. At $\eta=0$, it reduces to Stouffer's method applied to all validation values. We test $H_{0j}^\eta$ using $Q_j^\eta$ at the assigned level $\alpha_j$. Budgets are computed by integer arithmetic for the declared rational thresholds, preventing floating-point rounding from changing a boundary label.

The online multiple-testing procedure is SAFFRON \citep{ramdas2018saffron}, with $\lambda=0.5$, $w_0=\alpha/2=0.05$, and $\gamma_k=k^{-1.6}/\sum_{l\ge1}l^{-1.6}$. Let $C_{a:b}=\sum_{i=a}^b\1\{Q_i^\eta\le\lambda\}$, and let $\tau_1^R,\tau_2^R,\ldots$ be preceding rejection indices. Its update is
\begin{align}
\alpha_j=\min\bigg\{\lambda,(1-\lambda)\bigg[&w_0\gamma_{j-C_{1:j-1}}
 +(\alpha-w_0)\gamma_{j-\tau_1^R-C_{\tau_1^R+1:j-1}}\1\{\tau_1^R<j\}\nonumber\\
 &+\alpha\sum_{k\ge2:\tau_k^R<j}\gamma_{j-\tau_k^R-C_{\tau_k^R+1:j-1}}\bigg]\bigg\}.\label{eq:saffron}
\end{align}
Terms involving nonexistent rejections are omitted. The sequence is not renormalized to a finite candidate horizon. Every completed test updates the same history, while external request boundaries reset only the detector, not the online testing history. Different $\eta$ values have separate online testing runs. Under the theorem's independence assumptions, decreasing prior $Q$ values cannot reduce the assigned levels: eligibility counts decrease gamma indices, and new rejections add nonnegative rewards. This is the monotone SAFFRON version used in the FDR theorem.

Splitting uses odd positions for alarms and localization and even positions for validation, with the same raw-token window and start grids as \method{}. Both methods use Stouffer combinations for alarms and partial-conjunction validation and maintain their own persistent SAFFRON histories. Sums normalize by the number of observed coordinates in their own half. If the full passage has length $n$, at most $b=\lfloor\eta n\rfloor$ of its positions may be watermarked under the null, and all could be in validation. Thus a validation half of size $n_{\sf val}$ must delete $b$, not $\lfloor\eta n_{\sf val}\rfloor$, values. If $b\ge n_{\sf val}$, the validation $p$-value is one. The zero reports seen at $\eta=0.5$ in these runs describe this validation rule and the observed passage lengths.

\subsection{e-process restart and finite episodes}\label{app:su}
The e-process restart comparator is the average of the adaptive-log and online Grenander e-processes of \citet{su2026watermark}. Their pivotal statistic is $Y_t$, and our implementation uses the corresponding $p$-value $P_t=1-Y_t$. The adaptive-log component chooses $\beta_t\in[0,1/2]$ from preceding observations to maximize $\sum_{s<t}\log\{1-\beta-\beta\log P_s\}$, with $\beta_1=0$, then multiplies the current factor $1-\beta_t-\beta_t\log P_t$. The Grenander component fits a decreasing density MLE to preceding $p$-values plus one pseudo-observation at 1, retaining multiplicities, and evaluates it at the current value before adding that value to the fitting sample. The two complete processes are averaged, rather than averaging their per-token factors.

Each episode starts at capital one and ends at an eligible crossing of $1/\alpha$ or at the request end, with $\alpha=0.1$. A crossing reports the complete episode and restarts the component processes, which are also reset at the next request. In real-text experiments and the matched synthetic layout study, a crossing is eligible only after 64 tokens, and both component processes update at every token. Terminal episodes without an eligible crossing have decision zero. ESR includes every completed episode whose full interval satisfies $H_1^\eta$.

e-process restart receives the same raw values and request boundaries at all $\eta$. The threshold changes only the null or nonnull labels used for evaluation, never its behavior. Its original episode-level presence guarantee is not a proportion guarantee. Under a single entirely unwatermarked continuous stream, any report occurs if and only if the first e-process has an eligible crossing, so its probability is at most $\alpha$. Across $M$ independent unwatermarked requests, if the within-request crossing probability is $q$, then the deployment any-report probability is $1-(1-q)^M$, which need not be at most $\alpha$. This distinguishes repeated external opportunities from optional stopping within one valid episode.

Raw-p restart applies the same alarm and localization rules to unsplit $p$-values. It reports every candidate passage produced by an alarm, with no validation test or online multiplicity adjustment. Its passages remain fixed across $\eta$. 

\subsection{Interval labels and evaluation}
In the experiments, intervals use zero-based half-open coordinates $[l,r)$, corresponding to token positions $l+1,\ldots,r$. The detectors do not use the activation masks. These labels are used only to evaluate their outputs. The generator uses the mask to activate the target watermark decoder. Each candidate contributes once irrespective of how many watermarked positions or onsets it contains. The interval label uses watermarked count strictly greater than the integer budget. Report unions deduplicate positions within a deployment, and no truth-based deletion changes the actual report sequence.

\subsection{Aggregation and zero-denominator conventions}\label{app:metric-conventions}
A deployment is the complete stream or sequence of requests over which testing history persists. Let $U_d$ be its report union and $\mathcal M_d$ its watermarked positions. We report the mean deployment FDP and pooled power,
\begin{equation}
\widehat\FDR_\eta=\overline\FDP_\eta,\qquad
\widehat{\mathrm{Pow}}_\eta=\frac{\sum_{d,j}D_{dj}^\eta R_{dj}}{\sum_{d,j}D_{dj}^\eta},
\label{eq:empirical-metrics}
\end{equation}
and the token recovery and composition of the report union,
\begin{equation}
\mathrm{Cov}_{\rm tok}=\frac{\sum_d|U_d\cap\mathcal M_d|}{\sum_d|\mathcal M_d|},\qquad
\mathrm{Purity}=\frac{\sum_d|U_d\cap\mathcal M_d|}{\sum_d|U_d|}.
\label{eq:recovery-metrics}
\end{equation}
Power is the empirical counterpart of~\eqref{eq:conditional-power-target}, with a method-specific nonnull candidate denominator. Raw-p restart has no validation stage, so its power is inapplicable. ESR uses complete nonnull restart episodes, including unalarmed terminal episodes.

For an individual deployment with report union $U_d$, define
\begin{equation}
\mathrm{Cov}_d=\frac{|U_d\cap\mathcal M_d|}{|\mathcal M_d|},\qquad
\mathrm{Purity}_d=\frac{|U_d\cap\mathcal M_d|}{|U_d|},\qquad
U_d^+=\bigcup_{j:R_{dj}D_{dj}^\eta=1}B_{dj}.\label{eq:deployment-coverage}
\end{equation}
The first ratio is undefined if the deployment has no watermarked token, and the second is undefined if it has no report. Main tables do not average these individual ratios: they use the pooled counts in~\eqref{eq:recovery-metrics}. Qualified coverage is $\sum_d|U_d^+\cap\mathcal M_d|/\sum_d|\mathcal M_d|$ (Table~\ref{tab:qualified-coverage}). When the pooled nonnull candidate count is zero, power is zero for \method{} and Splitting. Inapplicable power, including that of Raw-p restart, remains dashed, and zero denominators for the other ratios also remain dashed. An unwatermarked deployment adds zero to the coverage numerator and denominator, but its false reports still enter FDR and the pooled purity denominator. A watermarked deployment with no report contributes its watermarked tokens to the coverage denominator. No-report FDP is always zero.

\begin{table}[!htb]\centering\small
\caption{\textbf{Qualified coverage, Cov.$^+$, for Tables~\ref{tab:synthetic} and~\ref{tab:model-generated}.}}\label{tab:qualified-coverage}
\begin{tabular*}{\linewidth}{@{\extracolsep{\fill}}l*{5}{c}@{}}\toprule
& \multicolumn{3}{c}{\textbf{$K=16$}} & \multicolumn{2}{c}{\textbf{Model-generated text}}\\
\cmidrule(lr){2-4}\cmidrule(lr){5-6}
Method & $\eta=0$ & $\eta=0.25$ & $\eta=0.5$ & Mistral-7B & Qwen3-8B\\\midrule
\multicolumn{6}{c}{\textit{Contiguous}}\\[2pt]
\textbf{FissionGuard} & 0.986 & 0.966 & 0.871 & 0.970 & 0.977\\
Splitting & 0.994 & 0.955 & 0.000 & 0.961 & 0.966\\
e-process & 1.000 & 0.891 & 0.874 & 0.985 & 0.989\\
Raw-p & 0.999 & 0.992 & 0.974 & 0.983 & 0.985\\
\addlinespace[3pt]
\multicolumn{6}{c}{\textit{Random}}\\[2pt]
\textbf{FissionGuard} & 0.973 & 0.922 & 0.003 & 0.952 & 0.955\\
Splitting & 0.973 & 0.660 & 0.000 & 0.877 & 0.911\\
e-process & 1.000 & 0.895 & 0.853 & 0.964 & 0.972\\
Raw-p & 0.992 & 0.985 & 0.953 & 0.963 & 0.965\\
\addlinespace[3pt]
\multicolumn{6}{c}{\textit{Gapped}}\\[2pt]
\textbf{FissionGuard} & 0.985 & 0.942 & 0.726 & 0.933 & 0.949\\
Splitting & 0.991 & 0.909 & 0.000 & 0.850 & 0.861\\
e-process & 1.000 & 0.862 & 0.855 & 0.964 & 0.973\\
Raw-p & 0.999 & 0.985 & 0.950 & 0.966 & 0.968\\
\addlinespace[3pt]
\multicolumn{6}{c}{\textit{Alternating}}\\[2pt]
\textbf{FissionGuard} & 0.965 & 0.803 & 0.622 & 0.955 & 0.957\\
Splitting & 0.969 & 0.686 & 0.000 & 0.938 & 0.953\\
e-process & 0.998 & 0.885 & 0.656 & 0.964 & 0.972\\
Raw-p & 0.998 & 0.986 & 0.682 & 0.967 & 0.969\\
\bottomrule\end{tabular*}
\end{table}

In the supplemental tables, $n_1$ counts nonnull candidate passages and $n_{\rm ep}$ counts complete nonnull restart episodes. Counts pool executions. At $\eta=0$, any candidate containing a watermarked position enters the power denominator, even if it contains no region onset. Report unions count each token once.

\paragraph{Restart episode success.}\label{app:eprocess-esr}
For e-process restart and Raw-p restart, let $E_{dj}$ be each complete monitoring episode, from the request start or the token after the previous alarm to the next eligible alarm or request end. Let $S_{dj}=1$ if the episode ends with an eligible e-process threshold crossing or Raw-p scan alarm, respectively, and $S_{dj}=0$ otherwise. All nonempty unalarmed terminal episodes are retained, including those shorter than the minimum scan length. The pooled episode success rate (ESR) is
\[
\mathrm{ESR}_\eta=\frac{\sum_{d,j}S_{dj}\1\{|E_{dj}\cap\mathcal M_d|>\lfloor\eta|E_{dj}|\rfloor\}}{\sum_{d,j}\1\{|E_{dj}\cap\mathcal M_d|>\lfloor\eta|E_{dj}|\rfloor\}}.
\]
Changing $\eta$ changes these truth labels without changing either restart method's episodes or reports. Truth labels use the complete episode $E_{dj}$; for Raw-p, FDR, coverage and purity still use its localized reports. A zero ESR denominator remains dashed. ESR and power among nonnull candidates use different rules to define their candidates. The main tables place the latter before the slash and ESR after it in the Power/ESR column. Every nonterminal episode ends with a report. Failures in ESR arise from unalarmed nonnull terminal episodes. Neither metric measures the fraction of all watermark regions found.
The power bounds in Appendix~\ref{app:power} do not apply to ESR. The supplementary records provide episode and crossing counts for each source group. Raw-p ESR is reported in the main and supplemental tables.

\section{Additional experimental results}\label{app:results}
The results below examine partial-conjunction calibration, unwatermarked controls, additional synthetic settings, and the numerical power bound. Appendix~\ref{app:real-extension} describes the model-generated text. The supplementary records provide the full numerical results and deployment counts.

\subsection{Matched layout comparison}\label{app:matched-layouts}
At watermarked positions, the token-level $p$-values follow $P_t\sim\operatorname{Beta}(1,K)$ for $K\in\{16,64\}$, and all remaining token-level $p$-values are independent uniforms. Proportion labels use each method's actual candidate passage, not the designed region.

Table~\ref{tab:synthetic} compares three proportion thresholds at fixed strength, and the supplementary records give the intermediate threshold $\eta=0.1$. Table~\ref{tab:synthetic-strength} gives results at $K=64$ for the same three thresholds. Within each strength, the threshold evaluations reuse the same 200 independently generated streams per layout and the same fission noise. They represent paired evaluations of 800 source streams, rather than additional independent deployments. \method{} and Splitting recompute their partial-conjunction tests and online testing histories at each threshold. The restart methods retain the same episodes or candidates and change only the null or nonnull labels used for evaluation.

\begin{table}[!htb]\centering
\caption{\textbf{Signal-strength comparison, $K=64$, $\alpha=0.1$, $\eta\in\{0,0.25,0.5\}$.} Each layout--strength cell has 200 independent continuous deployments. Every 11,008-token stream contains four 512-token regions with exactly 384 watermarked tokens each, separated by 256 unwatermarked tokens, with 4,096 unwatermarked tokens on either side. Larger $K$ strengthens the distribution of token-level $p$-values at watermarked positions $\operatorname{Beta}(1,K)$. All methods have 64-token minimum alarm spacing. Dashes and bold follow the conventions in Table~\ref{tab:synthetic}.}\label{tab:synthetic-strength}
\fontsize{8}{9.5}\selectfont
\setlength{\tabcolsep}{0.6pt}
\renewcommand{\arraystretch}{1.06}
\begin{tabular*}{\linewidth}{@{\extracolsep{\fill}}l@{\hspace{4pt}}*{5}{r}@{\hspace{4pt}}*{5}{r}@{\hspace{4pt}}*{5}{r}@{}}\toprule
 & \multicolumn{5}{c}{\textbf{$\eta=0$}} & \multicolumn{5}{c}{\textbf{$\eta=0.25$}} & \multicolumn{5}{c}{\textbf{$\eta=0.5$}}\\
\cmidrule(lr){2-6}\cmidrule(lr){7-11}\cmidrule(lr){12-16}
Method & FDR & \shortstack{Power/\\ESR} & Cov. & Cov.$^+$ & Purity & FDR & \shortstack{Power/\\ESR} & Cov. & Cov.$^+$ & Purity & FDR & \shortstack{Power/\\ESR} & Cov. & Cov.$^+$ & Purity\\\midrule
\multicolumn{16}{c}{\textit{Contiguous}}\\[2pt]
\textbf{FissionGuard} & 0.0443 & \textbf{0.977}/-- & 0.993 & 0.993 & 0.871 & 0.0008 & \textbf{0.944}/-- & \textbf{0.976} & \textbf{0.976} & 0.959 & 0.0004 & \textbf{0.974}/-- & \textbf{0.962} & \textbf{0.962} & \textbf{0.971}\\[1pt]
Splitting & 0.0386 & 0.972/-- & 0.996 & 0.996 & \textbf{0.880} & 0.0000 & 0.908/-- & 0.971 & 0.971 & \textbf{0.970} & 0.0000 & 0.000/-- & 0.000 & 0.000 & --\\[1pt]
e-process & 0.0038 & --/\textbf{1.000} & \textbf{1.000} & \textbf{1.000} & 0.223 & 0.1563 & --/1.000 & 1.000 & 0.930 & 0.223 & 0.1657 & --/1.000 & 1.000 & 0.926 & 0.223\\[1pt]
Raw-p & 0.2190 & --/1.000 & 1.000 & 1.000 & 0.686 & 0.2805 & --/1.000 & 1.000 & 0.990 & 0.686 & 0.3283 & --/1.000 & 1.000 & 0.977 & 0.686\\[1pt]
\addlinespace[3pt]
\multicolumn{16}{c}{\textit{Random}}\\[2pt]
\textbf{FissionGuard} & 0.0351 & \textbf{0.984}/-- & 0.985 & 0.985 & 0.680 & 0.0007 & \textbf{0.954}/-- & \textbf{0.967} & \textbf{0.967} & 0.729 & 0.0000 & \textbf{0.623}/-- & \textbf{0.613} & \textbf{0.613} & \textbf{0.767}\\[1pt]
Splitting & 0.0324 & 0.977/-- & 0.979 & 0.979 & \textbf{0.682} & 0.0000 & 0.840/-- & 0.875 & 0.875 & \textbf{0.747} & 0.0000 & 0.000/-- & 0.000 & 0.000 & --\\[1pt]
e-process & 0.0029 & --/\textbf{1.000} & \textbf{1.000} & \textbf{1.000} & 0.221 & 0.1212 & --/1.000 & 1.000 & 0.934 & 0.221 & 0.1613 & --/1.000 & 1.000 & 0.908 & 0.221\\[1pt]
Raw-p & 0.1746 & --/1.000 & 0.995 & 0.995 & 0.567 & 0.2242 & --/1.000 & 0.995 & 0.986 & 0.567 & 0.2716 & --/1.000 & 0.995 & 0.967 & 0.567\\[1pt]
\addlinespace[3pt]
\multicolumn{16}{c}{\textit{Gapped}}\\[2pt]
\textbf{FissionGuard} & 0.0391 & \textbf{0.974}/-- & 0.992 & 0.992 & 0.805 & 0.0017 & \textbf{0.892}/-- & \textbf{0.956} & \textbf{0.956} & 0.918 & 0.0010 & \textbf{0.946}/-- & \textbf{0.919} & \textbf{0.919} & \textbf{0.944}\\[1pt]
Splitting & 0.0317 & 0.962/-- & 0.994 & 0.994 & \textbf{0.820} & 0.0000 & 0.831/-- & 0.943 & 0.943 & \textbf{0.941} & 0.0000 & 0.000/-- & 0.000 & 0.000 & --\\[1pt]
e-process & 0.0037 & --/\textbf{1.000} & \textbf{1.000} & \textbf{1.000} & 0.222 & 0.2619 & --/1.000 & 1.000 & 0.911 & 0.222 & 0.2689 & --/1.000 & 1.000 & 0.908 & 0.222\\[1pt]
Raw-p & 0.1943 & --/1.000 & 1.000 & 1.000 & 0.648 & 0.3041 & --/1.000 & 1.000 & 0.982 & 0.648 & 0.3902 & --/1.000 & 1.000 & 0.956 & 0.648\\[1pt]
\addlinespace[3pt]
\multicolumn{16}{c}{\textit{Alternating}}\\[2pt]
\textbf{FissionGuard} & 0.0391 & \textbf{0.987}/-- & 0.990 & 0.990 & 0.679 & 0.0015 & \textbf{0.853}/-- & \textbf{0.912} & \textbf{0.911} & 0.782 & 0.0015 & \textbf{0.956}/-- & \textbf{0.673} & \textbf{0.672} & \textbf{0.996}\\[1pt]
Splitting & 0.0356 & 0.980/-- & 0.992 & 0.992 & \textbf{0.686} & 0.0000 & 0.498/-- & 0.684 & 0.684 & \textbf{0.983} & 0.0000 & 0.000/-- & 0.000 & 0.000 & --\\[1pt]
e-process & 0.0025 & --/\textbf{0.999} & \textbf{1.000} & \textbf{1.000} & 0.221 & 0.1435 & --/1.000 & 1.000 & 0.928 & 0.221 & 0.4998 & --/1.000 & 1.000 & 0.685 & 0.221\\[1pt]
Raw-p & 0.1734 & --/1.000 & 1.000 & 1.000 & 0.570 & 0.2237 & --/1.000 & 1.000 & 0.990 & 0.570 & 0.6176 & --/1.000 & 1.000 & 0.678 & 0.570\\[1pt]
\bottomrule\end{tabular*}
\end{table}

Table~\ref{tab:synthetic-strength} matches region length, marking count, evidence law and detector settings across four layouts. Every 11,008-token deployment has four 512-token regions beginning at zero-based positions $4096$, $4864$, $5632$ and $6400$. Consecutive regions are separated by 256 unwatermarked tokens, and the prefix and suffix each contain 4,096 unwatermarked tokens. There are no internal request boundaries. Each region contains 384 watermarked positions. Relative to a region, Contiguous marks $[64,448)$. Random samples 384 of the 512 positions uniformly without replacement. Gapped marks $[0,192)\cup[320,512)$, and Alternating marks $[128k,128k+96)$ for $k=0,1,2,3$. All masks precede and are independent of evidence and fission noise. The deterministic layouts place equal numbers of watermarked positions in the two fixed splitting halves.

Each layout--strength cell contains 200 independently simulated deployments with one fission randomization per deployment. Masks, evidence and noise use separate randomization streams. The two strengths share underlying uniform draws and fission noise within each layout, giving paired strength comparisons. Alarm and localization settings were fixed using separate pilot simulations before evaluation. Results are reported separately for each cell, and evaluations at different proportion thresholds reuse the same source streams.

For every method, an alarm is eligible only after at least 64 tokens have been processed since the most recent restart. The e-process updates on every token before an eligible crossing. No reset occurs at a latent marking transition. Uncrossed terminal episodes remain in its candidate records. Passages localized after an alarm may be shorter than 64 tokens. Appendix~\ref{app:synthetic-grid} uses the same default alarm settings and distinguishes continuous monitoring from experiments with declared external requests. Its two continuous Random cells at watermark proportion $0.75$ reuse the corresponding executions.

\subsection{Partial-conjunction calibration}\label{app:calibration}
The fixed-interval experiment evaluates partial-conjunction $p$-values under the proportion null. It uses lengths $n\in\{32,64,128\}$ and watermarked counts $b-1,b,b+1,b+\lceil0.1n\rceil$, retaining valid distinct counts in $[0,n]$, where $b=\lfloor\eta n\rfloor$. For each count, it considers $K=4$ and $64$, and contiguous marking versus an arrangement that assigns watermarked positions to the Splitting validation half first. Each synthetic cell has 10,000 repetitions at level $0.1$. Table~\ref{tab:boundary} reports the largest observed rejection rate among null cells at level $0.1$, separately for each $(n,\eta)$ and method. The table displays all three lengths. Maxima are taken over valid null watermarked counts, both $K$ values, and both arrangements. These are diagnostics, not simultaneous confidence limits.

A test for watermark presence can reject the null on nearly every window satisfying the proportion null when that window contains the watermarked positions permitted by the proportion null. The partial-conjunction test remains near or below its nominal level. A separate calibration check assigns arbitrarily strong validation evidence to exactly $b$ positions, leaving $n-b$ iid unwatermarked validation scores. At level $0.1$, observed rejection rates range from $0.0943$ to $0.1028$. This directly checks calibration when the allowed watermarked positions provide arbitrarily strong evidence. These fixed-interval experiments evaluate the validation statistic, whereas the online experiments additionally include alarms, localization and SAFFRON.

\begin{table}[!htb]\centering\small\setlength{\tabcolsep}{2.8pt}
\caption{\textbf{Fixed-interval proportion-null calibration at level $0.1$.} Maximum observed rejection rate over valid null watermarked counts, $K=4/64$, and two marking arrangements, with 10,000 repetitions per cell. Pres.: test for watermark presence, FG: \method{} and Split.: Splitting.}\label{tab:boundary}
\begin{tabular*}{\linewidth}{@{\extracolsep{\fill}}r*{9}{r}@{}}
\toprule
 & \multicolumn{3}{c}{$n=32$} & \multicolumn{3}{c}{$n=64$} & \multicolumn{3}{c}{$n=128$}\\
\cmidrule(lr){2-4}\cmidrule(lr){5-7}\cmidrule(l){8-10}
$\eta$ & Pres. & FG & Split. & Pres. & FG & Split. & Pres. & FG & Split.\\\midrule
0.00 & 0.1033 & 0.1033 & 0.1005 & 0.1033 & 0.1033 & 0.1021 & 0.1039 & 0.1039 & 0.1056\\
0.10 & 0.3401 & 0.0581 & 0.0885 & 0.4858 & 0.0458 & 0.0893 & 0.6869 & 0.0340 & 0.0765\\
0.25 & 0.8734 & 0.0341 & 0.0868 & 0.9854 & 0.0172 & 0.0810 & 1.0000 & 0.0087 & 0.0717\\
0.50 & 1.0000 & 0.0162 & 0.0000 & 1.0000 & 0.0068 & 0.0000 & 1.0000 & 0.0016 & 0.0000\\
\bottomrule\end{tabular*}
\end{table}

\subsection{Complete synthetic grid and unwatermarked controls}\label{app:synthetic-grid}
The complete synthetic grid uses the settings in Appendix~\ref{app:implementation} and contains 32 configurations with 7,600 deployments. It varies continuous or request-based monitoring, $K\in\{16,64\}$, watermark proportion $f\in\{0.5,0.75,0.875,1\}$, and region length $L\in\{128,256,512\}$. Four regions each contain exactly $fL$ watermarked positions, chosen uniformly without replacement, with separate comparisons using a contiguous unwatermarked segment or evenly spaced watermarked positions. In the request protocol, each region is one request, with three intervening unwatermarked requests and $H$ on either side, each of length 256. Continuous monitoring removes these boundaries. The base setting uses $L=512$ and $H=16$, and the request-exposure comparison crosses $H\in\{16,64,256\}$ with the three lengths at $f=0.75,K=64$.

Each configuration uses 200 deployments, except the $L=512$, $H\in\{16,256\}$ mixed request configurations at $K=64,f=0.75$ and the four entirely unwatermarked controls, which use 400 each. The two continuous Random configurations at $f=0.75$ reuse Tables~\ref{tab:synthetic} and~\ref{tab:synthetic-strength}. All other token-level $p$-values are independently simulated. The supplementary records give the complete results and candidate passages.

Lower watermark proportions, weaker signals and shorter regions reduce recovery. Denser marking narrows the coverage difference, and Splitting can recover more watermarked tokens in the dense request-based configurations.

\subsection{Numerical evaluation of the power bound}\label{app:theory-diagnostic}
We evaluate Theorem~\ref{thm:main-power} on a controlled selector whose deterministic interval family can be enumerated. Each of the 2,000 deployments per configuration has four requests of length $2N$. At each endpoint, an alarm is raised when the combined alarm $p$-value over the full request, computed using Stouffer's method, is at most $0.001$. The localization rule maximizes the positive standardized alarm sum over starts $s=0,\ldots,N$, giving suffix lengths $N$ through $2N$. The last fraction $f\in\{0.75,1\}$ of each request is watermarked. A suffix of length $n$ contains $m=\min(n,2Nf)$ watermarked positions. All possible candidates satisfy the tested alternative, so the expected number of nonnull candidates outside the family is zero. Validation uses the partial-conjunction statistic and persistent SAFFRON. This endpoint rule differs from the alarm rule used in the main comparisons.

We evaluate all 16 prespecified configurations. Table~\ref{tab:theory} compares the retained-subset union bound~\eqref{eq:homogeneous-tail} and the capped-moment bound~\eqref{eq:homogeneous-cap-tail}, each used in Theorem~\ref{thm:main-power}. At most four candidates occur, giving the deterministic floor $0.025\,4^{-1.6}/\zeta(1.6)>0.00119$, and both evaluations use $u=0.00119$. The expected nonnull candidate count is at least four times the full-request alarm probability, bounded below using the untrimmed presence Chernoff bound at alarm level $0.001$. The four-request factor cancels between numerator and denominator. Observed candidate counts are used only for empirical power.

\paragraph{Conservative moment evaluation.}
For uniform-$K$ Gumbel-max, $A=-\Phi^{-1}(P)$ has distribution function $\Phi(a)^K$ and density at most $K\phi(a)$. The function $a\mapsto H_{1/2}(c,\theta;a/\sqrt2)$ is nonincreasing. On $[-12,12]$, we use bin probabilities computed from the exact CDF and left-endpoint values on 98,304 bins to upper-bound its integral, while right endpoints provide a lower check. Outside this interval, the inequality
\[
H_{1/2}(c,\theta;a/\sqrt2)\le e^{-\theta a/\sqrt2+\theta^2/4}+e^{-\theta c}
\]
and the density bound give explicit Gaussian exponential-tail upper bounds. We add floating-point slack $10^{-10}(1+M)$ to an evaluated moment upper sum $M$. Untrimmed moments use the same construction without the constant capped term. A doubled grid and independent adaptive integration check the moment bounds. These are conservative numerical evaluations, not interval-arithmetic certificates. All reported three-decimal lower bounds are unchanged by grid refinement.

The union evaluation uses $\theta\in\{0.05,0.10,\ldots,4\}$. For the capped evaluation, each configuration uses one fixed pair $(c,\theta)$ for all its possible suffixes, selected from deterministic model calculations and recorded in the supplement. Every such pair yields a valid upper bound, regardless of numerical optimizer convergence. The moment bounds are then recomputed conservatively at that fixed pair. Neither parameter selection nor evaluation reads simulated validation values, selected means, realized levels, or outcomes. The resulting interval failure bounds are summed over all possible suffixes before dividing by the deterministic lower bound on the expected nonnull candidate count.

\begin{table}[!htb]\centering\small\setlength{\tabcolsep}{3.2pt}
\caption{\textbf{Power bounds for the endpoint selector.} Four requests per deployment and 2,000 original deployments per row. $f$ is the full-request watermark proportion. Union and trimmed-sum lower bounds use Theorem~\ref{thm:main-power}. Positive values are rounded down to three decimals and values numerically indistinguishable from one are displayed as $0.999$. Empirical power uses the pooled nonnull candidate count shown in the last column.}\label{tab:theory}
\begin{tabular}{@{}rrrrrrrr@{}}
\toprule $K$ & $N$ & $f$ & $\eta$ & Union bound & Trimmed sum & Empirical power & Nonnull $n$\\\midrule
64 & 32 & 1.00 & 0.25 & 0.994 & 0.999 & 1.0000 & 8000\\
64 & 128 & 1.00 & 0.25 & 0.999 & 0.999 & 1.0000 & 8000\\
64 & 512 & 1.00 & 0.25 & 0.999 & 0.999 & 1.0000 & 8000\\
64 & 32 & 1.00 & 0.50 & 0.000 & 0.000 & 1.0000 & 8000\\
64 & 128 & 1.00 & 0.50 & 0.689 & 0.999 & 1.0000 & 8000\\
64 & 512 & 1.00 & 0.50 & 0.999 & 0.999 & 1.0000 & 8000\\
64 & 32 & 0.75 & 0.25 & 0.000 & 0.991 & 1.0000 & 8000\\
64 & 128 & 0.75 & 0.25 & 0.000 & 0.999 & 1.0000 & 8000\\
64 & 512 & 0.75 & 0.25 & 0.999 & 0.999 & 1.0000 & 8000\\
64 & 32 & 0.75 & 0.50 & 0.000 & 0.000 & 0.9998 & 8000\\
64 & 128 & 0.75 & 0.50 & 0.000 & 0.000 & 1.0000 & 8000\\
64 & 512 & 0.75 & 0.50 & 0.000 & 0.999 & 1.0000 & 8000\\
4 & 128 & 1.00 & 0.25 & 0.000 & 0.000 & 0.9998 & 8000\\
4 & 512 & 1.00 & 0.25 & 0.000 & 0.999 & 1.0000 & 8000\\
64 & 32 & 0.75 & 0.00 & 0.999 & 0.999 & 1.0000 & 8000\\
64 & 128 & 0.75 & 0.00 & 0.999 & 0.999 & 1.0000 & 8000\\
\bottomrule\end{tabular}\end{table}

The capped-moment evaluation gives positive lower bounds in 12 of the 16 configurations, compared with eight for the union evaluation. For $(K,N,f,\eta)=(64,32,0.75,0.25)$ the lower bound changes from $0$ to $0.991$. For $(4,512,1,0.25)$ and $(64,512,0.75,0.5)$ it changes from $0$ to $0.999$ at the displayed precision. Four configurations still have zero evaluated lower bounds despite high empirical power. The capped-score bound avoids the retained-subset factor, while the sum over possible candidate intervals and the conservative SAFFRON floor remain. The supplement contains the fixed parameters, moment bounds and refinement checks, the sums over candidate intervals for all 16 configurations, the empirical counts, and three reference calculations.

\section{Model-generated text}\label{app:real-extension}
Proportion evaluations require token-level $p$-values, per-token target-watermark activation labels, and a declared sequence or request layout. The four-layout corpora contain 1,200 distinct source texts: 600 Contiguous texts across three temperatures and 600 internally mixed texts. Table~\ref{tab:model-generated} uses 800 of these texts, with 100 per model--layout pair at the specified model-specific temperature. The all-temperature Contiguous comparisons are reported separately below. The 300 all-unwatermarked controls are evaluated separately. Source inclusion does not depend on signal strength or detector output. These corpora were used during method development and are not held-out test data.

\subsection{Generation settings and unwatermarked controls}\label{app:model-protocol}
The Contiguous family uses ten prompts, five keys, three temperatures and two layouts, giving 300 source texts per model. Each text has 4,096 tokens and is one request. The duration layout has watermarked lengths $(32,64,128,256,512,1024)$ at zero-based starts $(128,288,480,736,1120,1760)$. The gap layout has seven 192-token watermarked passages starting at $(256,464,688,944,1264,1712,2416)$, with gaps $(16,32,64,128,256,512)$. All remaining positions are unwatermarked. Each model--temperature group has 100 texts and 168,000 watermarked tokens.

\paragraph{Internally mixed continuations.}
Random, Gapped and Alternating use 4,096-token autoregressive continuations with four regions of length $L\in\{128,256\}$, separated by 256 unwatermarked tokens. Region starts are $(1408,1792,2176,2560)$ for $L=128$ and $(1152,1664,2176,2688)$ for $L=256$. Each region has $3L/4$ watermarked positions. Random samples them without replacement. Relative to each region, Gapped leaves $[3L/8,5L/8)$ unwatermarked, and Alternating marks $[kL/4,kL/4+3L/16)$ for $k=0,1,2,3$. All other positions are unwatermarked. Activation masks are fixed independently before generation and define the token labels. Unwatermarked sampling does not consult the current watermark randomness.

The internally mixed generations use temperature 1.5 for Mistral-7B-Instruct-v0.3 and 2.0 for Qwen3-8B, chosen during method development based on earlier observations of watermark evidence strength. Each model has ten prompts and five keys, giving 50 sources per length--layout configuration and 300 overall, including the original 24 pilot texts once. When repetition induces dependence among pivotal statistics \citep{cai2025collision}, goodness-of-fit tests calibrated under an i.i.d.\ null can reject without watermark evidence. We favor higher-temperature generation to reduce repetition. These experiments evaluate watermark detection rather than text quality.

All methods use the same token-level $p$-values and external request boundaries, with one fission randomization per source. FDR averages deployment FDP, while the other metrics pool interval and token counts as in Appendix~\ref{app:metric-conventions}. Uncertainty estimates resample the ten prompt clusters, retaining all keys and randomizations within each cluster.

\paragraph{Key randomness in the implementation.}
The generator hashes a scheme-specific domain, watermark key, token position and preceding four token IDs to seed uniform scores over the vocabulary. Token positions continue through watermarked regions and unwatermarked gaps. Internally mixed generation also includes the source identifier in the domain. Contiguous generation uses a common domain and reuses keys across sources, so identical key, position and context can reproduce a seed. Token sampling at unwatermarked positions, activation masks and fission noise use separate randomization streams. The theoretical analysis models these pseudorandom variables under Assumption~\ref{ass:interface}.

\begin{figure}[!htb]
\centering
\includegraphics[width=\linewidth]{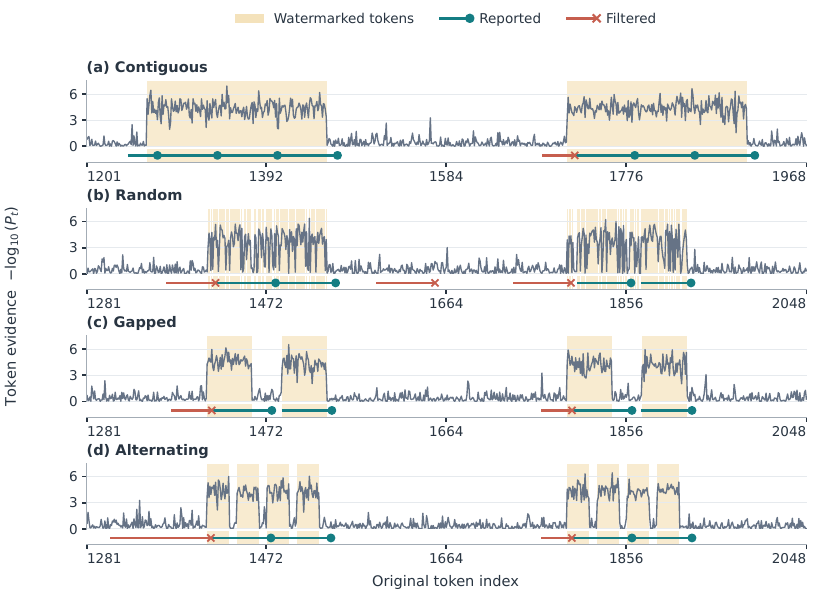}
\caption{\textbf{\method{} across four layouts.} Illustrative Mistral-7B-Instruct-v0.3 windows at temperature $1.5$, $\eta=0.25$ and $\alpha=0.1$, with 64-token minimum alarm spacing. Gray curves show token evidence $-\log_{10}(P_t)$, and amber shading marks watermarked tokens. Teal lines ending in circles denote reported passages, and coral lines ending in crosses denote filtered candidates. Every candidate intersecting each 768-token window is shown, with original token indices. Panel (a) contains fully watermarked 192-token regions. Panels (b)--(d) contain 128-token regions with $75\%$ watermarked positions, arranged randomly, around a central gap, or in alternating blocks. These windows illustrate the layouts, but their different region lengths and watermark proportions preclude a controlled comparison of recovery.}
\label{fig:four-layouts}
\end{figure}

Figure~\ref{fig:watermark-screening} shows tokens 1281--2048 of a Mistral Random text with $L=128$, including all seven candidates in that window. This source is included once in the Mistral results. The supplementary records also give the fixed-window model check on 300 watermarked/unwatermarked Mistral/Qwen source pairs, with 34,200 windows per method at level $0.1$ across lengths, watermarked counts, arrangements and one randomization.

Table~\ref{tab:ordinary} uses 100 entirely unwatermarked texts, with 50 from each model. The original control treats each 4,096-token text as one request. The exploratory replay retains the same token-level $p$-values in their original order and places request boundaries every 64, 128, 256, 512, 1,024, 2,048 or 4,096 tokens. \method{} and Splitting retain SAFFRON history across requests, while e-process restart resets at each boundary, following Appendix~\ref{app:implementation}. FDR is the frequency of any report within a complete deployment. Every report is false, so the restart rates apply to every $\eta$.

\begin{table}[!htb]\centering\small\setlength{\tabcolsep}{2pt}
\caption{\textbf{All-unwatermarked text: original requests and request replay, $\alpha=0.1$.} Each model has 50 original 4,096-token texts, at temperature 1.5 for Mistral and 2.0 for Qwen. Replay rows reuse the main-temperature sources with equally spaced request boundaries, preserving token order and total length. Entries are empirical FDR, equal here to the frequency of any report. \method{} uses one randomization per text.}\label{tab:ordinary}
\begin{tabular*}{\linewidth}{@{\extracolsep{\fill}}r*{9}{r}@{}}\toprule
& \multicolumn{3}{c}{FissionGuard} & \multicolumn{3}{c}{Splitting} & \multicolumn{3}{c}{e-process}\\
\cmidrule(lr){2-4}\cmidrule(lr){5-7}\cmidrule(lr){8-10}
Requests $\times$ tokens & $\eta=0$ & $0.25$ & $0.5$ & $\eta=0$ & $0.25$ & $0.5$ & $\eta=0$ & $0.25$ & $0.5$\\\midrule
\multicolumn{10}{c}{\textit{Mistral-7B}}\\[2pt]
$1\times4096$ & 0.0200 & 0.0000 & 0.0000 & 0.0200 & 0.0000 & 0.0000 & 0.0000 & 0.0000 & 0.0000\\
$2\times2048$ & 0.0200 & 0.0000 & 0.0000 & 0.0200 & 0.0000 & 0.0000 & 0.0200 & 0.0200 & 0.0200\\
$8\times512$ & 0.0200 & 0.0000 & 0.0000 & 0.0000 & 0.0000 & 0.0000 & 0.0800 & 0.0800 & 0.0800\\
$16\times256$ & 0.0000 & 0.0000 & 0.0000 & 0.0200 & 0.0000 & 0.0000 & 0.1400 & 0.1400 & 0.1400\\
$32\times128$ & 0.0000 & 0.0000 & 0.0000 & 0.0000 & 0.0000 & 0.0000 & 0.2600 & 0.2600 & 0.2600\\
\addlinespace[4pt]
\multicolumn{10}{c}{\textit{Qwen3-8B}}\\[2pt]
$1\times4096$ & 0.0400 & 0.0000 & 0.0000 & 0.0600 & 0.0000 & 0.0000 & 0.0000 & 0.0000 & 0.0000\\
$2\times2048$ & 0.0400 & 0.0000 & 0.0000 & 0.0600 & 0.0000 & 0.0000 & 0.0200 & 0.0200 & 0.0200\\
$8\times512$ & 0.0200 & 0.0000 & 0.0000 & 0.0200 & 0.0000 & 0.0000 & 0.1200 & 0.1200 & 0.1200\\
$16\times256$ & 0.0000 & 0.0000 & 0.0000 & 0.0200 & 0.0000 & 0.0000 & 0.2200 & 0.2200 & 0.2200\\
$32\times128$ & 0.0000 & 0.0000 & 0.0000 & 0.0200 & 0.0000 & 0.0000 & 0.3800 & 0.3800 & 0.3800\\
\bottomrule\end{tabular*}
\end{table}

e-process restart has low empirical FDR on the original single requests but exceeds the target at several replay lengths. \method{} and Splitting retain low empirical FDR. 

Exploratory 95\% intervals use 4,000 bootstrap resamples of the ten prompt clusters, keeping all keys, requests and randomizations together. The supplementary records give all six model--temperature groups and source-level counts.

\subsection{Temperature sensitivity}
Table~\ref{tab:temperatures} compares the four methods across generation temperatures on Contiguous text. \method{} has higher power among selected passages and token coverage than Splitting in all six model--temperature groups, with similar purity and low empirical FDR for both methods. Recovery is weakest at the lowest temperature, especially for Qwen. The restart methods recover more watermarked content but exceed the FDR target and yield lower-purity reports. Temperature alone does not ensure independent watermark evidence.

\begin{table}[!htb]\centering
\caption{\textbf{Contiguous streams across generation temperatures, $\alpha=0.1$, $\eta=0.25$.} Each model--temperature group contains 100 original 4,096-token texts: ten prompts, five keys and two layouts varying region lengths or gaps. All methods use the same texts and the default detection settings, with one fission randomization per text. Power and ESR use their respective method-specific nonnull interval denominators, and token metrics follow Appendix~\ref{app:metric-conventions}.}\label{tab:temperatures}
\fontsize{7}{8.5}\selectfont\setlength{\tabcolsep}{0.45pt}\renewcommand{\arraystretch}{1.06}
\begin{tabular*}{\linewidth}{@{\extracolsep{\fill}}l@{\hspace{4pt}}*{5}{c}@{\hspace{4pt}}*{5}{c}@{\hspace{4pt}}*{5}{c}@{}}\toprule
 & \multicolumn{5}{c}{\textbf{Temperature 1.2}} & \multicolumn{5}{c}{\textbf{Temperature 1.5}} & \multicolumn{5}{c}{\textbf{Temperature 2.0}}\\
\cmidrule(lr){2-6}\cmidrule(lr){7-11}\cmidrule(lr){12-16}
Method & FDR & Power/ESR & Cov. & Cov.$^+$ & Purity & FDR & Power/ESR & Cov. & Cov.$^+$ & Purity & FDR & Power/ESR & Cov. & Cov.$^+$ & Purity\\\midrule
\multicolumn{16}{c}{\textit{Mistral-7B}}\\[2pt]
\textbf{FissionGuard} & 0.0003 & 0.679/-- & 0.698 & 0.698 & 0.960 & 0.0040 & 0.939/-- & 0.970 & 0.970 & 0.950 & 0.0049 & 0.961/-- & 0.977 & 0.976 & 0.948\\[1pt]
Splitting & 0.0000 & 0.450/-- & 0.478 & 0.478 & 0.979 & 0.0000 & 0.911/-- & 0.961 & 0.961 & 0.963 & 0.0004 & 0.924/-- & 0.967 & 0.967 & 0.964\\[1pt]
e-process & 0.1458 & --/1.000 & 1.000 & 0.950 & 0.620 & 0.1564 & --/1.000 & 1.000 & 0.985 & 0.621 & 0.1543 & --/1.000 & 1.000 & 0.989 & 0.622\\[1pt]
Raw-p & 0.1081 & --/1.000 & 0.949 & 0.940 & 0.828 & 0.1760 & --/1.000 & 0.997 & 0.983 & 0.826 & 0.1854 & --/1.000 & 0.999 & 0.985 & 0.821\\[1pt]
\addlinespace[3pt]
\multicolumn{16}{c}{\textit{Qwen3-8B}}\\[2pt]
\textbf{FissionGuard} & 0.0008 & 0.184/-- & 0.071 & 0.071 & 0.957 & 0.0022 & 0.837/-- & 0.742 & 0.742 & 0.957 & 0.0065 & 0.971/-- & 0.978 & 0.977 & 0.947\\[1pt]
Splitting & 0.0000 & 0.105/-- & 0.041 & 0.041 & 0.970 & 0.0003 & 0.788/-- & 0.712 & 0.712 & 0.972 & 0.0000 & 0.940/-- & 0.966 & 0.966 & 0.966\\[1pt]
e-process & 0.1396 & --/0.966 & 0.653 & 0.621 & 0.618 & 0.1469 & --/1.000 & 0.986 & 0.956 & 0.623 & 0.1572 & --/1.000 & 1.000 & 0.989 & 0.619\\[1pt]
Raw-p & 0.1815 & --/0.999 & 0.381 & 0.377 & 0.742 & 0.1632 & --/1.000 & 0.856 & 0.845 & 0.824 & 0.1883 & --/1.000 & 0.998 & 0.985 & 0.817\\[1pt]
\bottomrule\end{tabular*}
\end{table}

\section{Extension: distribution-free online monitoring}
\label{app:industrial-extension}
This appendix studies the presence-testing special case at $\eta=0$ for industrial measurements. Conformal ranks provide uniform evidence without a Gaussian sensor model. The controlled industrial replay satisfies the joint distribution of the $p$-values in~\eqref{eq:unwatermarked-input-law}, allowing the same fission and testing argument to apply. The experiment tests for departures from nominal operation at $\eta=0$, rather than a positive proportion threshold.

\subsection{Online conformal values under the null}
\label{app:online-conformal}
Suppose $X_1,X_2,\ldots$ are exchangeable under $H_0$. At time $t$, a permutation-equivariant nonconformity rule fitted to $(X_1,\ldots,X_t)$ assigns $\operatorname{score}_t(X_i)$ to each observation, with larger scores more unusual. With independent $\xi_t\sim\Unif(0,1)$, define the smoothed online conformal value
\begin{equation}
 P_t=\frac{\sum_{i=1}^t\1\{\operatorname{score}_t(X_i)>\operatorname{score}_t(X_t)\}+\xi_t\sum_{i=1}^t\1\{\operatorname{score}_t(X_i)=\operatorname{score}_t(X_t)\}}{t}.
 \label{eq:online-conformal-rank}
\end{equation}
The online conformal validity theorem \citep{vovk2003exchangeability,vovk2005algorithmic,shafer2008tutorial} gives independent uniform values under this homogeneous null:
\begin{equation}
 \PP(P_1\le u_1,\ldots,P_T\le u_T)=\prod_{t=1}^T u_t,
 \qquad (u_1,\ldots,u_T)\in[0,1]^T.
 \label{eq:conformal-iid}
\end{equation}
Permutation symmetry gives uniform reverse-time ranks conditional on unordered observations and later ranks, and independent smoothing handles ties. Iteration proves the product law. Reusing one split-conformal calibration sample does not generally provide this independence.

Exchangeability under the homogeneous null alone does not establish independence between null and nonnull $p$-values in streams with departures. The following replay draws each target and its reference blocks independently, conditional on fixed empirical distributions. Its validity concerns these controlled draws, not the sensor observations in their original time order.

\subsection{Industrial monitoring and results}
\label{app:industrial-workflow}
We evaluate controlled replays of HAI 21.03 \citep{hai2021dataset}. Each observation is an intact 16-second block of 79 sensor channels. A differenced-VAR scorer fitted on train1 is frozen. The block score averages its 16 residual scores. Complete blocks retain their internal sample order and source identifiers, while incomplete, mixed-label, or gap-crossing blocks are excluded when defining the pools. Test1 supplies development pools. Each of test2--5 supplies separate empirical laws $G_0$ (ordinary blocks) and $G_a$ (attack blocks). Define departures $G_1=(G_0+G_a)/2$ and $G_2=G_a$.

At each replay index, independently draw a target block $X_t$ from its scheduled law, 63 references $X_{t,1},\ldots,X_{t,63}$ from $G_0$, and an independent uniform $\xi_t$. For the frozen block score $a$, use
\begin{equation}
 P_t=\frac{\sum_{i=1}^{63}\1\{a(X_{t,i})>a(X_t)\}+\xi_t\bigl(1+\sum_{i=1}^{63}\1\{a(X_{t,i})=a(X_t)\}\bigr)}{64}.
 \label{eq:industrial-rank}
\end{equation}
Under $G_0$, exchangeability of the 64 scores gives an exact uniform rank conditional on the pools. Independent sampling of each target and its references makes null $p$-values independent of one another and of nonnull $p$-values, satisfying the joint distribution of the $p$-values in~\eqref{eq:unwatermarked-input-law}. Both target and references must follow the same nominal empirical law. Resampling with replacement preserves draw independence but creates no new physical measurements.

Each measurement block yields a score from the fixed scoring rule and a randomized rank, and data fission then supplies alarm and validation values for localization, testing and online FDR control. The industrial implementation uses a BJ-sum alarm with windows $\{4,8,16,32\}$ and boundary 4. The selected interval is the suffix that triggers the alarm, fixed before applying a Stouffer combination test to its reserved validation values at a level assigned by SAFFRON, with $\alpha=0.1$, $\rho=0.5$, $w_0=0.05$, and $\lambda=0.5$. Splitting scans the same window grid on its alarm half and confirms with its validation half, retaining original block indices, while Raw-BJ restart directly reports each unsplit alarm. This industrial comparator is distinct from e-process restart for watermark detection.

Each primary replay has 2,048 blocks, with $G_1$ on $[128,160)$, $G_2$ on $[320,352)$, and $G_0$ elsewhere (zero-based). Each of four source pools contributes 300 mixed deployments, giving 1,200 deployments and 2,400 departure regimes. The primary horizon was fixed before evaluation. The same uninterrupted executions also retain the 512-, 1,024-, and 4,096-block exposure checks in the source records. The interval null is that every target position in the reported interval follows $G_0$. A draw from the ordinary component of $G_1$ is still part of a nonnull regime. Evaluation never relabels it using the realized source block.

\begin{figure}[!htb]
\centering\includegraphics[width=\linewidth]{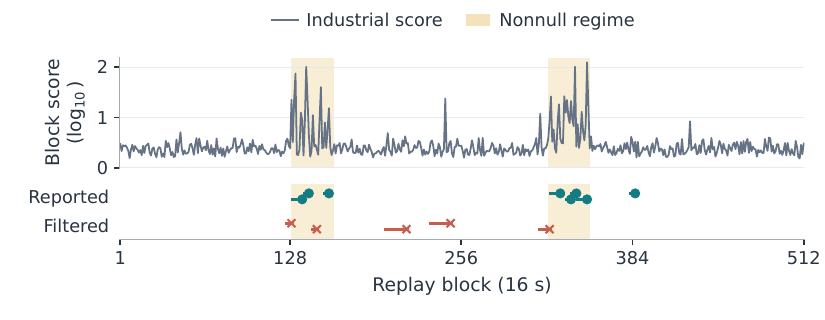}
\caption{\textbf{Industrial post-alarm screening.} First 512 blocks of the first 2,048-block test2 replicate, with original scores and all 13 candidate intervals in the crop. Shading marks the two departure regimes. Eight candidates are reported and five filtered, including two entirely null filters, but one false report remains. This example illustrates the procedure and does not estimate FDR.}
\label{fig:industrial-trace}
\end{figure}
\begin{table}[!htb]\centering\small\setlength{\tabcolsep}{5pt}
\caption{\textbf{Industrial monitoring over 2,048 blocks.} Mean FDP, mean false and true reports, pooled regime recall, and mean delay among detected regimes (replay seconds). All 1,200 mixed deployments enter FDR.}\label{tab:industrial-primary}
\begin{tabularx}{\linewidth}{@{}Yrrrrr@{}}\toprule
Method & FDR & Mean false & Mean true & Regime recall & Delay (s)\\\midrule
FissionGuard & 0.0484 & 0.40 & 6.53 & 0.831 & 167.7\\
Splitting & 0.0524 & 0.42 & 6.18 & 0.782 & 172.5\\
Raw-BJ restart & 0.3495 & 7.30 & 13.05 & 1.000 & 67.2\\
\bottomrule\end{tabularx}\end{table}
\begin{figure}[!htb]
\centering\includegraphics[width=\linewidth]{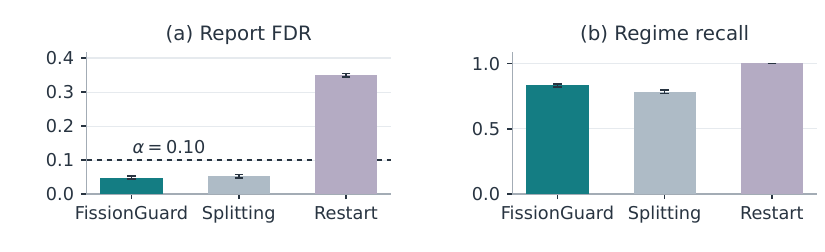}
\caption{\textbf{Industrial error control and recovery.} The 2,048-block comparison on 1,200 deployments: mean FDP and pooled regime recall at $\eta=0$, with 95\% paired within-pool bootstrap intervals. The dashed line marks $\alpha=0.1$. Restart denotes Raw-BJ restart without multiplicity correction.}
\label{fig:industrial-comparison}
\end{figure}
\enlargethispage{\baselineskip}
\method{} has lower empirical FDR than Raw-BJ restart and higher regime recall than Splitting (Table~\ref{tab:industrial-primary}). Confirmation filters ordinary activity while missing some departures and adding delay. Regime recall requires an intersecting report delivered during a departure, counting each regime once, and delay averages $16(t-s+1)$ replay seconds over detected regimes only. All reports enter FDP, including repeated true reports. The 4,000 paired within-pool bootstrap resamples quantify uncertainty conditional on these empirical pools. They do not represent independent industrial sites.

\end{document}